\documentclass[a4paper,10pt]{article}
\usepackage{graphicx}
\usepackage{amsmath}
\usepackage{citesort}
\usepackage{amsfonts}
\usepackage{amssymb}
\usepackage{amsmath}
\usepackage{latexsym}
\usepackage{srcltx}
\usepackage{color}
\usepackage{ntheorem}
\usepackage{braket}
\newcommand{\coneD}{{\mathrm D}}

\newcounter{mnotecount}[section]

\renewcommand{\themnotecount}{\thesection.\arabic{mnotecount}}

\newcommand{\tim}[1]{\mnote{{\bf tim:}#1}}

\newtheorem{theorem}{Theorem}[section]

\newtheorem{lemma}[theorem]{\sc Lemma\rm}

\newtheorem{proposition}[theorem]{\sc Proposition\rm}

\newtheorem{remark}[theorem]{\sc Remark\rm}

\newcommand{\ol}[1]{\overline{#1}{}}

\newcommand{\jlcax}[1]{}
\newcommand{\eean}{\nonumber\end{eqnarray}}

\newcommand{\kk}[1]{}%{\mnote{{\bf If we consider the KK case:} #1}}

\newcommand{\beq}{\begin{equation}}

\newcommand{\oY}{\overline Y}

\newcommand{\FS}       %{F_1} %
                  {F}
\newcommand{\HS} %{F_2}
       {H_{\mbox{\scriptsize volume}}}

{\ptc{this should be removed in the oberwolfach version}}%

\newcommand{\eeal}[1]{\label{#1}\end{eqnarray}}
\newcommand{\bed}{\begin{deqarr}}
\newcommand{\eed}{\end{deqarr}}
\newcommand{\bedl}[1]{\begin{deqarr}\label{#1}}
\newcommand{\eedl}[2]{\arrlabel{#1}\label{#2}\end{deqarr}}

\newcommand{\bel}[1]{\begin{equation}\label{#1}}
\newcommand{\bea}{\begin{eqnarray}}
\newcommand{\bean}{\begin{eqnarray}\nonumber}
\newcommand{\beal}[1]{\begin{eqnarray}\label{#1}}
\newcommand{\eea}{\end{eqnarray}}

\newcommand{\qed}{\hfill $\Box$ \medskip}

\newcommand{\be}{\begin{equation}}
\newcommand{\eeq}{\end{equation}}
\newcommand{\ee}{\end{equation}}
\newcommand{\beqa}{\begin{eqnarray}}
\newcommand{\eeqa}{\end{eqnarray}}
\newcommand{\beqan}{\begin{eqnarray*}}
\newcommand{\eeqan}{\end{eqnarray*}}
\newcommand{\ba}{\begin{array}}
\newcommand{\ea}{\end{array}}

\newcommand{\mcM}{{\mycal M}}
\newcommand{\mcD}{{\mycal D}}

\newcommand{\scri}{{\mycal I}}%

\newcommand{\mnote}[1]%{}
{\protect{\stepcounter{mnotecount}}$^{\mbox{\footnotesize
$%\!\!\!\!\!\!\,
\bullet$\themnotecount}}$ \marginpar{%\color{red}%
\raggedright\tiny\em
$\!\!\!\!\!\!\,\bullet$\themnotecount: #1} }

\newcommand{\warn}[1]%{}%{}
{\protect{\stepcounter{mnotecount}}$^{\mbox{\footnotesize
$%\!\!\!\!\!\!\,
\bullet$\themnotecount}}$ \marginpar{%\color{red}%
\raggedright\tiny\em
$\!\!\!\!\!\!\,\bullet$\themnotecount: {\bf Warning:} #1} }

\newcommand{\eq}[1]{(\ref{#1})}

\newcommand{\ptc}[1]{\mnote{{\bf ptc:}#1}}

\newcommand{\mcL}{{\mycal L}}

\newcommand{\beqar}{\begin{deqarr}}
\newcommand{\eeqar}{\end{deqarr}}

\newcommand{\beaa}{\begin{eqnarray*}}
\newcommand{\eeaa}{\end{eqnarray*}}

\DeclareFontFamily{OT1}{rsfs}{}
\DeclareFontShape{OT1}{rsfs}{m}{n}{ <-7> rsfs5 <7-10> rsfs7 <10-> rsfs10}{}
\DeclareMathAlphabet{\mycal}{OT1}{rsfs}{m}{n}

{\catcode `\@=11 \global\let\AddToReset=\@addtoreset}
\AddToReset{equation}{section}

{\catcode `\@=11 \global\let\AddToReset=\@addtoreset}
\AddToReset{figure}{section}

{\catcode `\@=11 \global\let\AddToReset=\@addtoreset}
\AddToReset{table}{section}

\begin{document}

%\changefont{phv}{m}{n}

\title{%
%KIDs like special cones even more%
KIDs prefer special cones%
\thanks{Preprint UWThPh-2013-31. }
}
\author{Tim-Torben Paetz%
\thanks{E-mail:  Tim-Torben.Paetz@univie.ac.at}  \vspace{0.5em}\\  \textit{Gravitational Physics, University of Vienna}  \\ \textit{Boltzmanngasse 5, 1090 Vienna, Austria }}

\maketitle

\vspace{-0.2em}

\begin{abstract}
As complement to \textit{Class.\ Quantum Grav.}\ \textbf{30} (2013) 235036
%(arXiv:1305.7468 [gr-qc])
we analyze Killing initial data on characteristic Cauchy surfaces in conformally
rescaled vacuum spacetimes satisfying Friedrich's conformal field equations.
As an application, we derive and discuss the KID equations
on a light-cone with vertex at past timelike infinity.
\end{abstract}

\noindent
\hspace{2.1em} PACs number: 04.20.Ex, 04.20.Ha

\tableofcontents

\section{Introduction}

Gaining a better insight into properties and peculiarities of spacetimes which represent (physically meaningful) solutions to Einstein's field equations belongs to the core of the analysis of general relativity.
One question of interest concerns the existence of spacetimes which possess certain symmetry groups, mathematically expressed via
a Lie algebra of Killing vector fields on that spacetime.
A fundamental issue in this context is to systematically construct such spacetimes in terms of an initial value problem.
By that it is meant to supplement the usual constraint equations, which need to be satisfied by a suitably
 specified set of initial data, by some further equations
which make sure that the emerging spacetime contains one or several Killing vector fields.
In vacuum, such Killing Initial Data (KIDs) are well-understood in the spacelike case as well as in the characteristic case, cf.\ \cite{beig,moncrief,CP1} and references therein.
In this article we would like to complement the analysis of the characteristic case given in \cite{CP1} to
spacetimes satisfying Friedrich's conformal field equations, and in particular to analyse
the case where the initial surface is a light-cone with vertex
at past timelike infinity.

In a first step, Section~\ref{sec_KID_unphys}, we translate the Killing equation into the unphysical, conformally rescaled spacetime.
The so-obtained ``unphysical Killing equations'' constitute the main focus of our subsequent analysis.
Assuming the validity of the conformal field equations, recalled in Section~\ref{sec_setting},
we shall derive necessary-and-sufficient conditions on a characteristic initial surface which guarantee
the existence of a vector field satisfying the unphysical Killing equations, cf.\ Theorem~\ref{KID_eqns_main}.%
\footnote{This issue has already been analysed in \cite{kannar2}. However, it is claimed there that regularity of the principal part of a wave equation suffices to
guarantee uniqueness of solutions, and counter-examples of this assertion can be easily constructed.
For instance, let $\Theta$ be the unique solution of the wave-equation $\Box_g \Theta =1$ which vanishes on the initial surface which we assume to be a light-cone $C_O$, i.e.\ $\Theta|_{C_O}=0$. Then $\Theta|_{I^+(O)}>0$, at least sufficiently close to $O$. Consider the non-regular wave-equation $\Box_g f - \frac{1}{\Theta}f=0$. For given initial data $f|_{C_O}=0$ there exist at least 3 solutions: $f=0,\pm \Theta$.
}
In Section~\ref{KID_4dim} we then restrict attention to four spacetime dimensions (it will be indicated
that the higher dimensional case is more intricate). As in \cite{CP1} we shall see that many of the conditions
obtained in Section~\ref{sec_KID_unphys} are automatically satisfied.
The remaining ``KID equations'' are collected in Theorem~\ref{KID_eqns_main_cone2} (cf.\ Proposition~\ref{KID_eqns_main_cone3}) for a light-cone,
and %, since this case can be treated without any additional efforts,
in Theorem~\ref{KID_eqns_main_hzpe2} for two characteristic hypersurfaces intersecting transversally.

In Section~\ref{sec_special_cone} we then apply Theorem~\ref{KID_eqns_main_cone2} to the ``special cone'' $C_{i^-}$
whose vertex is located at past timelike infinity (assuming the cosmological constant to be zero). As for ``ordinary cones'' treated in \cite{CP1} it turns out
that some of the KID equations determine a class of candidate fields on the initial surface while the remaining
``reduced KID equations'' provide restrictions on the initial data to make sure that one of these candidate fields does
indeed extend to a spacetime vector field satisfying the unphysical Killing equations.
However, contrary to the ``ordinary case'', and this explains our title,
 on $C_{i^-}$ the candidate fields can be explicitly computed, and, besides, the
reduced KID equations can be given in terms of explicitly known quantities.
The main result for the $C_{i^-}$-cone is the contents of Theorem~\ref{KID_eqns_main_cone_infinity}.

Finally, in Appendix~\ref{app_Fuchsian} we recall a result on Fuchsian ODEs which will be of importance in the main part,
in Appendix~\ref{app_conformal} we review conformal Killing vector fields on the round 2-sphere.

\section{Setting}
\label{sec_setting}

Our analysis will be carried out in the so-called unphysical spacetime $(\mcM,g,\Theta)$, related to the physical spacetime $(\tilde{\mcM\enspace}\hspace{-0.5em},\tilde g)$, $\tilde g$ being a solution to  Einstein's field equations, via a conformal rescaling,
\begin{equation*}
 \tilde g \overset{\phi}{\mapsto} g:=\Theta^2 \tilde g\;, \quad  \tilde{\mcM\enspace}\hspace{-0.5em} \overset{\phi}{\hookrightarrow} \mcM\;,
  \quad \Theta|_{\phi( \tilde{\mcM\enspace}\hspace{-0.5em})}>0
 \;.
\end{equation*}
The part of $\partial \phi( \tilde{\mcM\enspace}\hspace{-0.5em})$ on which the conformal factor $\Theta$ vanishes represents
``infinity'' in the physical spacetime.

In $(\mcM,g,\Theta)$ Einstein's vacuum field equations with cosmological constant $\lambda$ are replaced by Friedrich's conformal field equations (cf.\ e.g.\ \cite{F3}),
which read  in $d\geq 4$ spacetime dimensions
\begin{eqnarray}
 && \nabla_{\rho} d_{\mu\nu\sigma}{}^{\rho} =0\;,
 \label{conf1}
\\
 && \nabla_{\mu} L_{\nu\sigma} - \nabla_{\nu}L_{\mu\sigma} = \Theta^{d-4}\nabla_{\rho}\Theta \, d_{\nu\mu\sigma}{}^{\rho}\;,
 \label{conf2}
\\
 && \nabla_{\mu}\nabla_{\nu}\Theta = -\Theta L_{\mu\nu} + s g_{\mu\nu}\;,
 \label{conf3}
\\
 && \nabla_{\mu} s = -L_{\mu\nu}\nabla^{\nu}\Theta\;,
 \label{conf4}
\\
 && (d-1)(2\Theta s - \nabla_{\mu}\Theta\nabla^{\mu}\Theta )= \lambda \;,
 \label{conf5}
\\
 && R_{\mu\nu\sigma}{}^{\kappa}[ g] = \Theta^{d-3} d_{\mu\nu\sigma}{}^{\kappa} + 2(g_{\sigma[\mu} L_{\nu]}{}^{\kappa}  - \delta_{[\mu}{}^{\kappa}L_{\nu]\sigma} )
 \label{conf6}
\;,
\end{eqnarray}
with $g_{\mu\nu}$, $\Theta$, $s$, $ L_{\mu\nu}$ and $d_{\mu\nu\sigma}{}^{\rho}$ regarded as unknowns.
The trace of \eq{conf3} can be read as the definition of the function $s$,
\begin{equation}
s := \frac{1}{d}g^{\mu\nu}\nabla_{\mu}\nabla_{\nu}\Theta + \frac{1}{2d(d-1)} R\Theta
\label{dfn_s}
 \;.
\end{equation}
The tensor field $L_{\mu\nu}$ is the Schouten tensor
\begin{equation}
  L_{\mu\nu} := \frac{1}{d-2}R_{\mu\nu} - \frac{1}{2(d-1)(d-2)} R g_{\mu\nu}
 \;,
\label{dfn_L}
\end{equation}
while
\begin{equation}
d_{\mu\nu\sigma}{}^{\rho} := \Theta^{3-d}C_{\mu\nu\sigma}{}^{\rho}
\label{dfn_d}
\end{equation}
 is a rescaling of the conformal Weyl tensor $C_{\mu\nu\sigma}{}^{\rho}$.

The conformal field equations are equivalent to the vacuum Einstein equations where $\Theta$ is positive, but remain regular even where $\Theta$ vanishes.
The Ricci scalar $R$ turns out to be a conformal gauge source function which reflects the freedom to choose the conformal factor $\Theta$.
It can be prescribed arbitrarily.
% we will impose $R=R_c=\mathrm{const}$ henceforth.
%(Usually $R=0$ is the most convenient choice, however, for later applications we prefer
%to choose a gauge which is flexible enough to be compatible with $\Theta=1$, which makes it possible
%to include the physical, non-rescaled spacetimes, as well.)

In \eq{conf1}-\eq{conf6} the fields $s$, $ L_{\mu\nu}$ and $d_{\mu\nu\sigma}{}^{\rho}$ are treated as independent of   $g_{\mu\nu}$ and $\Theta$.
However, once a solution has been constructed they are related to   $g_{\mu\nu}$ and $\Theta$ via \eq{dfn_s}-\eq{dfn_d}.
When talking about a solution of the conformal field equations we therefore just need to specify the pair $(g_{\mu\nu},\Theta)$.

The conformal field equations imply a wave equation for the Schouten tensor, which will be of importance later on.
It can be derived as follows: One starts by taking the divergence of \eq{conf2}. Using then \eq{conf1}, \eq{conf3}, \eq{conf6} and the tracelessness of the rescaled
Weyl tensor one finds  (cf.\ \cite{ttp} where the $4$-dimensional case is treated in detail),
\begin{eqnarray*}
 \Box_g L_{\mu\nu} &=&  - 2 R_{\mu}{}^{\alpha}{}_{\nu}{}^{\beta}L_{\alpha\beta} +g_{\mu\nu}|L|^2
 +\frac{1}{2(d-1)} \nabla_{\mu}\nabla_{\nu}R
  + \frac{1}{d-1}RL_{\mu\nu}
\\
&&
 + (d-4)\Big[L_{\mu}{}^{\alpha}L_{\nu\alpha}
+ \Theta^{d-5}\nabla_{\alpha}\Theta\nabla_{\beta}\Theta
 d_{\mu}{}^{\alpha}{}_{\nu}{}^{\beta}\Big]
\;,
\end{eqnarray*}
with $|L|^2 := L_{\mu}{}^{\nu}L_{\nu}{}^{\mu}$.
Supposing that $\Theta$ has no zeros, or that we are in the 4-dimensional case, we can use \eq{conf2} to rewrite this as
\begin{eqnarray}
 \Box_g L_{\mu\nu} &=&   - 2 R_{\mu}{}^{\alpha}{}_{\nu}{}^{\beta}L_{\alpha\beta} +g_{\mu\nu}|L|^2
 +\frac{1}{2(d-1)} \nabla_{\mu}\nabla_{\nu}R
  + \frac{1}{d-1}RL_{\mu\nu}
\nonumber
\\
&&
 + (d-4)\Big[L_{\mu}{}^{\alpha}L_{\nu\alpha}
+2\Theta^{-1}\nabla^{\alpha}\Theta \nabla_{[\alpha}L_{\mu]\nu}\Big]
 \;.
\label{wave_schouten}
\end{eqnarray}

\section{KID equations in the unphysical spacetime}
\label{sec_KID_unphys}

\subsection{The Killing equation in terms of a conformally rescaled metric}

\begin{lemma}
A vector field $\tilde X$ is a Killing vector field
in the physical spacetime $(\tilde{\mcM\enspace}\hspace{-0.5em},\tilde g)$ if and only if
its push-forward $X:= \phi_*\tilde X$ is a
conformal Killing vector field in the unphysical spacetime $(\mcM,g ,\Theta )$
%, at least in $(\phi(\tilde{\mcM\enspace}\hspace{-0.5em}),g)$,
 and satisfies there the equation $ X^{\kappa}\nabla_{\kappa}\Theta =\frac{1}{d}\Theta \nabla_{\kappa} X^{\kappa}$.
\end{lemma}
\begin{proof}
 By definition $\tilde X^{\mu}$ is a Killing field if and only if (set $\tilde X_{\mu}:=\tilde  g_{\mu\nu}\tilde X^{\nu}$ and $X_{\mu}:= g_{\mu\nu}X^{\nu}$)
 \begin{eqnarray}
  &&\tilde\nabla_{(\mu}\tilde  X_{\nu)}=0
   \nonumber
  \\
    \Longleftrightarrow &&\tilde\nabla_{(\mu}(\Theta^{-2} X_{\nu)}) =0
   \nonumber
  \\
  \Longleftrightarrow &&  \nabla_{(\mu}(\Theta^{-2} X_{\nu)})   + 2\Theta^{-2}X_{(\mu}\nabla_{\nu)}\log\Theta
   = g_{\mu\nu}\Theta^{-2}X_{\kappa}\nabla^{\kappa}\log\Theta
      \nonumber
    \\
  \Longleftrightarrow &&\nabla_{(\mu}X_{\nu)}
   = g_{\mu\nu}\Theta^{-1} X_{\kappa}\nabla^{\kappa}\Theta
      \nonumber
     \\
  \Longleftrightarrow &&  \nabla_{(\mu}X_{\nu)} =\frac{1}{d}\nabla_{\kappa} X^{\kappa}\,g_{\mu\nu}
   \quad \& \quad X^{\kappa}\nabla_{\kappa}\Theta =\frac{1}{d}\Theta \nabla_{\kappa} X^{\kappa}
   \label{2_conditions}
 \end{eqnarray}
(note that $\Theta|_{\phi(\tilde{\mcM\enspace}\hspace{-0.5em})}>0$).
 \qed
\end{proof}

\begin{remark}
{\rm
The conditions \eq{2_conditions}, which replace the Killing equation in the unphysical spacetime, make sense also where $\Theta$ is vanishing, supposing that
$g$  can be smoothly extended across $\{ \Theta=0\}$ (note that the conformal Killing equation induces a linear symmetric hyperbolic system of propagation equations for $\phi_*\tilde X$ which implies that $\phi_*\tilde X$ is smoothly extendable
across the conformal boundary \cite{F_lambda}).
}
\end{remark}

\begin{remark}
{\rm
We will refer to \eq{2_conditions} as the \textit{unphysical Killing equations}.
}
\end{remark}

The main object of this work is to extract necessary-and-sufficient conditions on a characteristic initial surface which ensure the existence
of some vector field $X$ which fulfills the unphysical Killing equations, so that its pull-back is a Killing vector field of the physical spacetime.

\subsection{Necessary conditions for the existence of Killing vector fields}

Let us first derive some implications of the unphysical Killing equations \eq{2_conditions} under the hypothesis that the conformal field
equations \eq{conf1}-\eq{conf6} are satisfied.
%The idea is to derive a sufficient number of necessary conditions which then turn out to be sufficient, as well.

From the conformal Killing equation we first derive a system of wave equations for $X$ and for the function
\begin{equation}
Y:=\frac{1}{d}\nabla_{\kappa}X^{\kappa}
\end{equation}
(set $\Box_g:=g^{\mu\nu}\nabla_{\mu}\nabla_{\nu}$):
\begin{eqnarray}
\Box_g X_{\mu} + R_{\mu}{}^{\nu}X_{\nu}  + (d-2)\nabla_{\mu}Y &=&0
 \label{wave_X0}
\;,
\\
 \Box_g Y+ \frac{1}{d-1}\Big(\frac{1}{2}X^{\mu}\nabla_{\mu}R+ R Y\Big)   &=&0
 \label{wave_Y0}
 \;.
\end{eqnarray}
With \eq{2_conditions} and \eq{dfn_s}
we find
\begin{eqnarray}
 0 &=& \Box_g(X^{\mu}\nabla_{\mu}\Theta -\Theta Y)
 \nonumber
 \\
 &\equiv& \Box_g X^{\mu} \nabla_{\mu}\Theta
  + X^{\mu}\nabla_{\mu}\Box_g \Theta
   + X^{\mu} R_{\mu}{}^{\nu}\nabla_{\nu}\Theta
    +2\nabla_{\nu}X_{\mu}\nabla^{\mu}\nabla^{\nu}\Theta
\nonumber
    \\
 && - Y \Box_g \Theta-\Theta \Box_g Y-2\nabla_{\mu}\Theta \nabla^{\mu}Y
 \nonumber
 \\
 &=&  d(X^{\mu}\nabla_{\mu}s    + sY - \nabla_{\mu}\Theta \nabla^{\mu}Y)
 \;.
  \label{waveeqn_XnablaTheta0}
\end{eqnarray}
We set
\begin{equation}
 A_{\mu\nu} :=  \nabla_{\mu}X_{\nu} + \nabla_{\nu} X_{\mu} - 2 Y g_{\mu\nu}
 \;.
\end{equation}
Using the second Bianchi identity, \eq{dfn_L}, \eq{wave_X0} and \eq{wave_Y0} we obtain
\begin{eqnarray}
 \hspace{-1em}\Box_gA_{\mu\nu} &\equiv & 2\nabla_{(\mu}\Box_g X_{\nu)}
  + 2R_{\kappa(\mu}\nabla^{\kappa}X_{\nu)}
   - 2R_{\mu}{}^{\alpha}{}_{\nu}{}^{\beta}A_{\alpha\beta}
- 4 R_{\mu\nu}Y
 \nonumber
\\
 && +  2X^{\kappa} \nabla_{(\mu}R_{\nu)\kappa} - 2X^{\kappa}\nabla_{\kappa}R_{\mu\nu}
 -2\Box_g Y g_{\mu\nu}
 \nonumber
\\
 &=& 2R_{(\mu}{}^{\kappa}A_{\nu)\kappa} - 2R_{\mu}{}^{\alpha}{}_{\nu}{}^{\beta}A_{\alpha\beta}
  -2(d-2)( \mcL_XL_{\mu\nu} + \nabla_{\mu}\nabla_{\nu}Y)\;,
      \label{waveeqn_A0}
\end{eqnarray}
(Recall that
\begin{eqnarray*}
 \mcL_X L_{\mu\nu} \equiv X^{\kappa}\nabla_{\kappa}L_{\mu\nu} +2 L_{\kappa(\mu}\nabla_{\nu)}X^{\kappa}
  \;.\;)
 \end{eqnarray*}
Hence the conformal Killing equation for $X$, which is $A_{\mu\nu}=0$,  implies
\begin{eqnarray}
 B_{\mu\nu} := \mcL_XL_{\mu\nu} +   \nabla_{\mu}\nabla_{\nu}Y  =0
 \;.
\end{eqnarray}

\subsection{KID equations on a characteristic initial surface}

\subsubsection{First main result}

We are now in a position to formulate our first main result.
Here and in the following we use an overbar to denote restriction to the initial surface.
\begin{theorem}
 \label{KID_eqns_main}
Assume we have been given,
in dimension $d\geq 4$,  an``unphysical'' spacetime $(\mcM, g, \Theta)$, with ($g, \Theta$) a smooth solution of the conformal field equations \eq{conf1}-\eq{conf6}.
Assume further that $\Theta$ is bounded away from zero if $d \geq 5$.
Consider some characteristic initial surface $N\subset \mcM$ (for definiteness we think of  a light-cone or two transversally intersecting null hypersurfaces).
 Then there exists a vector field %$\hat X$
satisfying the unphysical Killing equations \eq{2_conditions} on $\mathrm{D}^+(N)$
(i.e.\  representing a Killing field of the physical spacetime)
if and only if there exists a vector field $X$ and a function $Y$   which fulfill the following equations (recall the definitions \eq{dfn_s} and \eq{dfn_L} for $s$ and $L_{\mu\nu}$, respectively)
 \begin{enumerate}
  \item[(i)] $\Box_g X_{\mu} + R_{\mu}{}^{\nu}X_{\nu}  + (d-2)\nabla_{\mu}Y= 0$,
  \item[(ii)] $ \Box_g Y+ \frac{1}{d-1}\big(\frac{1}{2}X^{\mu}\nabla_{\mu}R+ R Y\big) =0$,
   \item[(iii)] $ \ol\phi = 0$ with $ \phi :=  X^{\mu}\nabla_{\mu}\Theta - \Theta Y  $,
  \item[(iv)] $\ol\psi=0$ with
 $\psi := X^{\mu}\nabla_{\mu}s    +sY -\nabla_{\mu}\Theta \nabla^{\mu}Y $,
  \item[(v)] $\ol A_{\mu\nu} =0$ with
 $A_{\mu\nu} \equiv \nabla_{\mu}X_{\nu} + \nabla_{\nu} X_{\mu} - 2Y g_{\mu\nu}$,
  \item[(vi)] $\breve{\ol B}_{\mu\nu}:=\ol B_{\mu\nu} - \frac{1}{d}\ol g_{\mu\nu}\ol B_{\alpha}{}^{\alpha}=0$ with
$B_{\mu\nu}\equiv \mcL_XL_{\mu\nu} +   \nabla_{\mu}\nabla_{\nu}Y$.
 \end{enumerate}
%Moreover, $\hat X=X$ and $\nabla_{\kappa}\hat X^{\kappa}= d\cdot Y$.
\end{theorem}

\begin{proof}
``$\Longrightarrow$'': Follows from the considerations above if one takes $X=\hat X$ and $Y=\frac{1}{d}\nabla_{\kappa}\hat X^{\kappa}$.

``$\Longleftarrow$'': We will derive a homogeneous system of wave equations from which we conclude the vanishing of $A_{\mu\nu}$ and $\phi$ as well as the relation $Y=\frac{1}{d}\nabla_{\kappa}X^{\kappa}$, and thus the validity of \eq{2_conditions}.
Since by assumption \eq{wave_X0} and \eq{wave_Y0} hold we can repeat the steps which led us to \eq{waveeqn_A0},
\begin{eqnarray}
 \Box_gA_{\mu\nu}  &=& 2R_{(\mu}{}^{\kappa}A_{\nu)\kappa} - 2R_{\mu}{}^{\alpha}{}_{\nu}{}^{\beta}A_{\alpha\beta}
  -2(d-2)B_{\mu\nu}
 \;.
      \label{waveeqn_A}
\end{eqnarray}
With (i), (ii) and the definition \eq{dfn_s} of $s$ we find
\begin{eqnarray}
 \Box_g\phi &\equiv& \Box_g X_{\mu} \nabla^{\mu}\Theta
  + X^{\mu}\nabla_{\mu}\Box_g \Theta
   + X^{\mu} R_{\mu}{}^{\nu}\nabla_{\nu}\Theta
    +2\nabla_{\mu}X_{\nu}\nabla^{\mu}\nabla^{\nu}\Theta
\nonumber
    \\
 && - Y \Box_g \Theta-\Theta \Box_g Y-2\nabla_{\mu}\Theta \nabla^{\mu}Y
 \nonumber
 \\
 &=& d\psi - \frac{1}{2(d-1)}R\phi+ A_{\mu\nu} \nabla^{\mu}\nabla^{\nu}\Theta
 \;.
  \label{waveeqn_XnablaTheta}
\end{eqnarray}
We use (i), (ii), \eq{dfn_s}, \eq{dfn_L} as well as the conformal field equations \eq{conf3} and \eq{conf4}
(which imply % $\Box_g\Theta =ds - \frac{1}{2(d-1)} R\Theta$ and
$\Box_g s = \Theta|L|^2 - \frac{1}{2(d-1)}(sR + \nabla^{\mu}\Theta\nabla_{\mu}R)$) to obtain
\begin{eqnarray}
 \Box_g\psi &\equiv&
 \Box_gX_{\mu}\nabla^{\mu}s
  + X^{\mu}\nabla_{\mu}\Box_gs +  X^{\mu}R_{\mu}{}^{\nu}\nabla_{\nu}s
 + A_{\mu\nu}\nabla^{\mu}\nabla^{\nu}s
 \nonumber
\\
 &&    +3Y\Box_gs +s\Box_gY+ 2\nabla_{\nu}s\nabla^{\nu}Y
  -\nabla_{\mu}\Box_g\Theta \nabla^{\mu}Y - 2R_{\mu}{}^{\nu}\nabla_{\nu}\Theta \nabla^{\mu}Y
 \nonumber
\\
&&  -\nabla^{\mu}\Theta\nabla_{\mu}\Box_gY
 - 2\nabla_{\mu}\nabla_{\nu}\Theta \nabla^{\mu}\nabla^{\nu}Y
 \nonumber
\\
  &=& |L|^2 \phi
 + A_{\mu\nu}(\nabla^{\mu}\nabla^{\nu}s - 2\Theta  L_{\kappa}{}^{\mu} L^{\nu\kappa}) + 2\Theta L^{\mu\nu} B_{\mu\nu}
\nonumber
\\
&& + \frac{1}{2(d-1)}\big(  A_{\mu\nu}\nabla^{\mu}R  \nabla^{\nu}\Theta   - \nabla^{\mu}R \nabla_{\mu}\phi - R\psi   \big)
 \;.
  \label{waveeqn_Xnabla_s}
\end{eqnarray}
%
%\begin{eqnarray*}
%\Box_g B_{\mu\nu} &\equiv&  2 L_{(\mu}{}^{\kappa}\nabla_{\nu)}\Box_g X_{\kappa}
%+ \Box_g X_{\kappa}\nabla^{\kappa}L_{\mu\nu}
%+ 4\nabla^{\sigma} L_{(\mu}{}^{\kappa}\nabla_{\nu)}\nabla_{\sigma}X_{\kappa}
%\\
%&&
%+2\Box_gL_{\kappa(\mu}\nabla_{\nu)}X^{\kappa}
%- \frac{1}{(d-1)(d-2)} R R_{\beta(\mu}\nabla^{\beta}X_{\nu)}
%\\
%&&
%+ X^{\kappa}\nabla_{\kappa}\Box_g  L_{\mu\nu}
%+ X^{\kappa}R_{\kappa\beta}\nabla^{\beta}L_{\mu\nu}
% - 2X^{\beta} L_{(\mu}{}^{\kappa} \nabla_{\beta}R_{\nu)\kappa}
%+ 2 X^{\beta}L_{(\mu}{}^{\kappa}\nabla_{\nu)}R_{\beta\kappa}
%\\
%&& + 4 X^{\kappa}\nabla^{\alpha}L_{\beta(\mu}R_{\nu)\kappa\alpha}{}^{\beta}
%\\
%&& + \nabla_{\mu}\nabla_{\nu}\Box_g Y
%+ 2R_{(\mu}{}^{\beta}\nabla_{\nu)}\nabla_{\beta}Y
% - R_{\mu}{}^{\alpha}{}_{\nu}{}^{\beta}\nabla_{\alpha}\nabla_{\beta}Y
%+ 2  \nabla_{(\mu}R_{\nu)\beta}\nabla^{\beta}Y
%\\
%&&
%- \nabla_{\beta}R_{\mu\nu}\nabla^{\beta}Y
%+2Y\Box_g L_{\mu\nu}
%+ \frac{2}{(d-1)(d-2)} RR_{\mu\nu} Y
% - \frac{2}{d-2}R_{\mu}{}^{\kappa} R_{\nu\kappa}Y
%\\
%&&
%+A^{\sigma\kappa} \nabla_{\sigma}\nabla_{\kappa}L_{\mu\nu}
%- 2 L_{(\mu}{}^{\kappa} R_{\nu)}{}^{\alpha}{}_{\kappa}{}^{\beta}A_{\alpha\beta}
%+ \frac{1}{d-2}R_{\mu}{}^{\alpha} R_{\nu}{}^{\beta}A_{\alpha\beta}
%\end{eqnarray*}
%
As a straightforward consequence of the first Bianchi identity we observe the identity
\begin{eqnarray}
 \frac{1}{2}\nabla_{\mu}A_{\nu\kappa}+ \nabla_{[\nu}A_{\kappa]\mu}
\,\equiv\, \nabla_{\mu} \nabla_{\nu} X_{\kappa}+  R_{\nu\kappa\mu}{}^{\alpha}X_{\alpha}  - 2\nabla_{(\mu}Y g_{\nu)\kappa}
+   \nabla_{\kappa}Y g_{\mu\nu}
 \;.
\label{relation_nabla_A}
\end{eqnarray}
Using this as well as the Bianchi identities a lengthy computation reveals that
\begin{eqnarray}
\Box_g B_{\mu\nu} &\equiv&
(\nabla_{(\mu}A_{|\alpha\beta|}  +  2\nabla_{[\alpha}A_{\beta](\mu})(2\nabla^{\alpha}L_{\nu)}{}^{\beta}
-\frac{1}{4(d-1)}\delta_{\nu}{}^{\alpha}\nabla^{\beta}R)
\nonumber
\\
&&\hspace{-3em}
+A^{\alpha\beta} (\nabla_{\alpha}\nabla_{\beta}L_{\mu\nu}
- 2 L_{(\mu}{}^{\kappa} R_{\nu)\alpha\kappa\beta}
+ \frac{1}{d-2}R_{\mu\alpha} R_{\nu\beta})
\nonumber
\\
&&\hspace{-3em}
- \frac{R^2}{4(d-1)^2(d-2)}A_{\mu\nu}
 + \nabla_{\mu}\nabla_{\nu}[\Box_g Y    + \frac{1}{d-1}(\frac{1}{2}X^{\kappa}\nabla_{\kappa}R+ R Y)]
\nonumber
\\
&&\hspace{-3em}
+(2 L_{(\mu}{}^{\kappa}\nabla_{\nu)} + \nabla^{\kappa}L_{\mu\nu})  [\Box_g X_{\kappa} + R_{\kappa}{}^{\alpha}X_{\alpha}  + (d-2)\nabla_{\kappa}Y ]
\nonumber
\\
&&\hspace{-3em}
+(X^{\kappa}\nabla_{\kappa}+2Y)\Box_g  L_{\mu\nu}
+2\Box_gL_{\kappa(\mu}\nabla_{\nu)}X^{\kappa}
  - \frac{1}{2(d-1)}X^{\kappa}\nabla_{\kappa}\nabla_{(\mu}\nabla_{\nu)}R
\nonumber
\\
&&\hspace{-3em}
- 2 L_{\kappa}{}^{\alpha}R_{(\mu}{}^{\kappa}\nabla_{\nu)} X_{\alpha}
  - 2X^{\beta} L_{(\mu}{}^{\kappa} \nabla_{\beta}R_{\nu)\kappa}
 -  2L_{\mu}{}^{\kappa} R_{\nu\kappa}Y
- 2R_{\mu}{}^{\alpha}{}_{\nu}{}^{\beta}\nabla_{\alpha}\nabla_{\beta}Y
\nonumber
\\
&&\hspace{-3em}
+ \frac{1}{d-1}[R L_{\mu\nu} Y-  RL_{\beta(\mu}\nabla^{\beta}X_{\nu)}
-Y\nabla_{\mu}\nabla_{\nu} R
 -\nabla_{(\mu}X^{\kappa}\nabla_{\nu)} \nabla_{\kappa}R]
\nonumber
\\
 &&\hspace{-3em}
 +2 (d-4)\nabla^{\alpha}Y (\nabla_{(\mu}L_{\nu)\alpha}  -\nabla_{\alpha}L_{\mu\nu} )
 \;.
\label{first_wave_B}
\end{eqnarray}
There is a second useful relation satisfied  by $A_{\mu\nu}$ which follows from the Bianchi identities,
\begin{eqnarray}
  &&\hspace{-3em}  2L^{\alpha\beta}(\nabla_{\beta}\nabla_{[\alpha}A_{\nu]\mu}  - \nabla_{\mu}\nabla_{[\alpha}A_{\nu]\beta})
 \nonumber
\\
&=&  2 L^{\alpha\beta}X^{\kappa}\nabla_{\kappa}R_{\alpha(\mu\nu)\beta}
+ L^{\alpha\beta}( R_{\beta\nu\mu}{}^{\kappa}A_{\alpha\kappa}
-  R_{\mu\beta\alpha}{}^{\kappa} A_{\nu\kappa}  )
\nonumber
\\
&& + 4L^{\alpha\beta}( R_{\alpha(\mu\nu)}{}^{\kappa}\nabla_{[\beta}X_{\kappa]}
+ R_{(\mu|\beta \alpha|}{}^{\kappa}\nabla_{\nu)} X_{\kappa} )
\nonumber
\\
&&  - 2g_{\mu\nu} L^{\alpha\beta}\nabla_{\alpha}\nabla_{\beta}Y
+ 4L_{(\mu}{}^{\beta}\nabla_{\nu)}\nabla_{\beta}Y
- \frac{1}{d-1}R\nabla_{\mu}\nabla_{\nu}Y
\:.
\label{relation_nabla_A2}
\end{eqnarray}
Employing (i), (ii), the conformal field equations, the wave equation for the Schouten tensor \eq{wave_schouten} as well as the identities \eq{relation_nabla_A} and \eq{relation_nabla_A2} we deduce from \eq{first_wave_B} in another tedious computation that
\begin{eqnarray}
\Box_g B_{\mu\nu} &\equiv&
 2(g_{\mu\nu}L^{\alpha \beta}-   R_{\mu}{}^{\alpha}{}_{\nu}{}^{\beta})B_{\alpha\beta}
 - 2 R_{(\mu}{}^{\kappa} B_{\nu)\kappa}
+  \frac{2}{d-1}RB_{\mu\nu}
\nonumber
\\
&&\hspace{-3em}
   + 2L^{\alpha\beta}(\nabla_{\beta}\nabla_{[\alpha}A_{\nu]\mu}  -  \nabla_{\mu}\nabla_{[\alpha}A_{\nu]\beta})
\nonumber
\\
&&\hspace{-3em}
+ (\nabla_{(\mu}A_{|\alpha\beta|}  +  2\nabla_{[\alpha}A_{\beta](\mu})(2\nabla^{\alpha}L_{\nu)}{}^{\beta}
-\frac{1}{4(d-1)}\delta_{\nu)}{}^{\alpha}\nabla^{\beta}R)
\nonumber
\\
&&\hspace{-3em}
+A^{\alpha\beta} [\nabla_{\alpha}\nabla_{\beta}L_{\mu\nu}
- 2 L_{(\mu}{}^{\kappa} R_{\nu)\alpha\kappa\beta}
+ 2L_{\mu\alpha} R_{\nu\beta}
+ L_{\alpha}{}^{\kappa} (2R_{\mu\beta\nu\kappa}+  R_{\nu\beta\mu\kappa})
\nonumber
\\
&&\hspace{-3em}
-  2g_{\mu\nu} L_{\alpha\kappa}L_{\beta}{}^{\kappa}]
 + |L|^2 A_{\mu\nu}
+  L^{\alpha\beta} R_{\mu\alpha\beta}{}^{\kappa} A_{\nu\kappa}
-  \frac{1}{d-1}RL_{(\mu}{}^{\kappa}A_{\nu)\kappa}
\nonumber
\\
&&\hspace{-3em}
+  (d-4)[2L_{(\mu}{}^{\beta} B_{\nu)\beta} - L_{\mu}{}^{\alpha}L_{\nu}{}^{\beta}A_{\alpha\beta}
-\Theta^{d-5}  \nabla_{\alpha}\Theta d_{\kappa\mu\nu}{}^{\alpha}    ( \nabla^{\kappa}\phi
-    \nabla_{\beta}\Theta A^{\kappa\beta} )]
\nonumber
\\
&&\hspace{-3em}
 -(d-4) \Theta^{d-6} \phi \nabla_{\alpha}\Theta \nabla_{\beta}\Theta d_{\mu}{}^{\alpha}{}_{\nu}{}^{\beta}
\nonumber
\\
&&\hspace{-3em}
+  (d-4)\Theta^{-1}\nabla^{\alpha}\Theta [  2\nabla_{[\alpha} B_{\mu]\nu} + (\nabla_{[\mu}A_{|\nu\kappa|}+2  \nabla_{[\nu}A_{\kappa][\mu} )L_{\alpha]}{}^{\kappa}]
 \;.
\label{waveeqn_B}
\end{eqnarray}
%
%\begin{eqnarray*}
% && X^{\kappa}\nabla_{\kappa}\Box_g  L_{\mu\nu}
%+2\Box_gL_{\kappa(\mu}\nabla_{\nu)}X^{\kappa}
%+2Y\Box_g L_{\mu\nu}
%\\
%&=& -2(d-4)  \nabla^{\kappa}Y \Theta^{d-4} \nabla_{\alpha}\Theta d_{\kappa(\mu\nu)}{}^{\alpha}
%\\
%&&
%-2(d-4) \Theta^{-2} \phi \nabla^{\alpha}\Theta \nabla_{[\alpha}L_{\mu]\nu}
%\\
%&& +  (d-4)\Theta^{-1}\nabla^{\alpha}\Theta L_{[\alpha}{}^{\kappa}[ \nabla_{\mu]}A_{\nu\kappa}+  \nabla_{\nu}A_{\mu]\kappa}  - \nabla_{\kappa}A_{\mu]\nu}]
%\\
%&& +  2(d-4)\Theta^{-1}\nabla^{\alpha}\Theta\nabla_{[\alpha} B_{\mu]\nu}
%\\
%&&+2(d-4)\Theta^{-1}\nabla_{\alpha}\Theta  \nabla_{[\mu} L_{\kappa]\nu}A^{\alpha\kappa}
%\\
%&&-2(d-4)\Theta^{-1}  \nabla_{[\mu} L_{\kappa]\nu}\nabla^{\kappa}\phi
%\end{eqnarray*}
%
While  the before-last line contains negative powers of $\Theta $  merely in 5 dimensions, the  last line contains such powers in any
dimension $d\geq 5$. This is the point where our assumption enters that  $\Theta$ is bounded away from zero for  $d\geq 5$, since this ensures that
\eq{waveeqn_B} is a regular equation also in higher dimensions.

Note that the right-hand side of \eq{waveeqn_B} involves second-order derivatives of $A_{\mu\nu}$, which is why we regard $\nabla_{\sigma}A_{\mu\nu}$
as another unknown for which we derive a wave equation. However, since the right-hand side of \eq{waveeqn_A} does not involve derivatives,
such a wave equation is easily obtained by differentiation (and, once again, the second Bianchi identity),
\begin{eqnarray}
\hspace{-2em}\Box_g\nabla_{\sigma}A_{\mu\nu}  &=&
  2\nabla_{\sigma}(R_{(\mu}{}^{\kappa}A_{\nu)\kappa}
 - R_{\mu}{}^{\alpha}{}_{\nu}{}^{\kappa}A_{\alpha\kappa}  )
   + 2A_{\alpha(\mu}(\nabla_{\nu)}R_{\sigma}{}^{\alpha} - \nabla^{\alpha}R_{\nu)\sigma} )
\nonumber
\\
 &&
  -4R_{\sigma\kappa(\mu}{}^{\alpha}\nabla^{\kappa}A_{\nu)\alpha}
+ R_{\alpha\sigma}\nabla^{\alpha}A_{\mu\nu}
  - 2(d-2)\nabla_{\sigma}B_{\mu\nu}.
      \label{waveeqn_nabla_A}
\end{eqnarray}
In the current setting the equations \eq{waveeqn_A}-\eq{waveeqn_Xnabla_s}, \eq{waveeqn_B} and \eq{waveeqn_nabla_A}
 form a closed homogeneous system of regular wave equations for $ A_{\mu\nu}$, $\phi$,  $\psi$, $B_{\mu\nu}$ and $\nabla_{\sigma}A_{\mu\nu}$.
The assumptions (iii)-(v) assure that the first three fields vanish initially.
By (ii) and (v) we have
\begin{equation*}
  \ol B_{\alpha}{}^{\alpha} = \frac{1}{2(d-1)} \ol{X^{\kappa}\nabla_{\kappa}R }+  2 \ol{L^{\mu\nu}\nabla_{\mu}X_{\nu}} + \ol{\Box_g Y} = \ol L^{\mu\nu}\ol A_{\mu\nu} =0
 \;,
\end{equation*}
which, together with (vi), implies
\begin{equation}
\ol B_{\mu\nu}=0
\;.
\label{vanishing_olB}
\end{equation}
It remains to verify the vanishing of $\ol{\nabla_{\sigma}A_{\mu\nu}}$.
This follows from Lemma~\ref{vanishing_nabla_A} below, together with (i), (ii), (v) and \eq{vanishing_olB}.
We thus have vanishing initial data for the homogeneous system of wave equations  \eq{waveeqn_A}-\eq{waveeqn_Xnabla_s}, \eq{waveeqn_B} and \eq{waveeqn_nabla_A}.
Due to standard uniqueness results for wave equations all the fields
involved need to vanish identically.

It is important to note that we had to treat $X$ and $Y$ as independent unknowns so far. The vanishing of $A_{\mu\nu}$ and $\phi$ implies that the unphysical Killing equations  \eq{2_conditions} hold for $X$ only once we have shown that $Y=\frac{1}{d}\nabla_{\kappa}X^{\kappa}$.
Fortunately we have
\begin{equation}
 0 = A_{\alpha}{}^{\alpha} = 2\nabla_{\kappa}X^{\kappa}-2dY
\;,
\end{equation}
and the theorem is proved.
\qed
\end{proof}

\subsubsection{Adapted null coordinates}

Before we state and prove Lemma~\ref{vanishing_nabla_A}, which is needed to complete the proof of Theorem~\ref{KID_eqns_main}, it is useful, also with regard to later purposes, to
introduce \textit{adapted null coordinates} on light-cones and on transversally intersecting null hypersurfaces. We will be rather sketchy here, the details can be found e.g.\ in \cite{CCM2,rendall}.

First let us consider a light-cone $C_O\subset \mcM$ with vertex $O\in \mcM$ in a $d=n+1$-dimensional spacetime $(\mcM,g)$.
 We use coordinates $(x^0=0,x^1=r,x^A)$, $A=2,\dots,n$, adapted to $C_O$ in the sense that $C_O\setminus\{O\} = \{ x^0=0\}$,
$r$ parameterizes the null geodesics generating the initial surface, and the $x^A$'s are local coordinates on the level sets $\{x^0=0,r=\mathrm{const}\}\cong S^{n-1}$.
On $C_O$ the metric then reads
\begin{equation*}
  \ol g = \ol g_{00} ( \mathrm{d}x^{0})^2 + 2\nu_0  \mathrm{d}x^{0} \mathrm{d}x^{1}
 + 2\nu_A  \mathrm{d}x^{0} \mathrm{d}x^{A} + \ol g_{AB}  \mathrm{d}x^{A} \mathrm{d}x^{B}
\;.
\end{equation*}
Note that these coordinates are singular at the vertex of the cone. Moreover, we stress that we do not impose any gauge condition off the cone.
The inverse metric takes the form
\begin{equation*}
 \ol g^{\sharp} = 2 \nu^0 \partial_{0} \partial_1   + \overline{g}^{11}\partial^{2}_{1} + 2 \overline{g}^{1A} \partial_1  \partial_{A}
+ \overline{g}^{AB}\partial_{A} \partial_{B}
\;,
\end{equation*}
with
\begin{equation*}
 \nu^0=(\nu_0)^{-1}\;, \quad \ol g^{1A} = -\nu^0\ol g^{AB}\nu_B\;, \quad \ol g^{11} = (\nu^0)^2(\ol g^{AB}\nu_A\nu_B - \ol g_{00})
 \;.
\end{equation*}
It is customary to introduce the following quantities:
\begin{eqnarray*}
 \chi_{A}{}^B &:=& \frac{1}{2}\ol g^{BC}\partial_1\ol g_{AC} \quad \text{null second fundamental form}
\;,
\\
 \tau &:=& \chi_A{}^A \quad \text{expansion}
\;,
\\
 \sigma_{A}{}^B &:=& \chi_A{}^B - \frac{1}{n-1}\tau \delta_A{}^B \quad \text{shear tensor}
\;.
\end{eqnarray*}

Next, let us consider two smooth hypersurfaces $N_a$, $a=1,2$, with transverse intersection along a smooth submanifold $S$.
Then, near the $N_a$'s one can introduce coordinates $(x^1,x^2,x^A)$, $A=3,\dots,n+1$, such that $N_a=\{ x^a=0\}$.
On $N_1$ the coordinate $x^2$ parameterizes the null geodesics
$\{x^1=0,x^A=\mathrm{const}^A\}$ generating $N_1$ and vice versa.
Since the hypersurfaces are required to be characteristic the metric takes there the specific form, on $N_1$ say,
\begin{equation*}
  g|_{N_1}= \ol g_{11} ( \mathrm{d}x^{1})^2 + 2\ol g_{12}  \mathrm{d}x^{1} \mathrm{d}x^{2}
 + 2\ol g_{1A}  \mathrm{d}x^{1} \mathrm{d}x^{A} + \ol g_{AB}  \mathrm{d}x^{A} \mathrm{d}x^{B}
 \;,
\end{equation*}
similarly on $N_2$. The quantities $\tau$, $\sigma_A{}^B$, $\chi_{A}{}^B$  are defined on $N_1$ and $N_2$ analogous to the light-cone-case.

\subsubsection{Some useful relations}

In this section we consider a light-cone. However, we note that exactly the same relations hold in
the case of two intersecting null hypersurfaces.

Recall that the wave equations for $X$ and $Y$, \eq{wave_X0} and \eq{wave_Y0}, imply a wave equation which is
satisfied by $A_{\mu\nu}$, namely \eq{waveeqn_A},
\begin{eqnarray}
 \Box_gA_{\mu\nu}  &=& 2 R_{(\mu}{}^{\kappa} A_{\nu)\kappa}
 - 2 R_{\mu}{}^{\alpha}{}_{\nu}{}^{\beta} A_{\alpha\beta}
  -2(d-2) B_{\mu\nu}
 \;.
   \label{waveeqn_Agen}
\end{eqnarray}
Furthermore, one straightforwardly verifies that in adapted null coordinates
\begin{eqnarray}
 \ol{\Box_g A_{\mu\nu}} &=& 2\nu^0(\ol{\nabla_1\nabla_0A_{\mu\nu}} + \ol R_{01(\mu}{}^{\alpha}\ol  A_{\nu)\alpha})
  + \ol g^{11} \ol{\nabla_1\nabla_1A_{\mu\nu}}
   \nonumber
  \\
&&  + 2\ol g^{1A} (\ol{\nabla_1\nabla_AA_{\mu\nu}}
 + \ol R_{A1(\mu}{}^{\alpha} \ol A_{\nu)\alpha})
  + \ol g^{AB}\ol{\nabla_A\nabla_BA_{\mu\nu}}
  \;.
  \label{identity_wave_Agen}
\end{eqnarray}
We equate the trace of \eq{waveeqn_Agen} on the initial surface with \eq{identity_wave_Agen}.
Making use of the formulae for the Christoffel symbols in adapted null coordinates in \cite[Appendix~A]{CCM2}, an elementary calculation yields the following set of equations where $f$, $f_A$ and $f_{AB}$ denote generic (multi-linear) functions which vanish whenever their arguments vanish.

\vspace{1em}

\noindent
\underline{$(\mu\nu)=(11)$:}
\begin{eqnarray}
 (\partial_1 + \frac{\tau}{2}- \ol\Gamma^0_{01} - 2\ol \Gamma^1_{11})\ol{\nabla_0A_{11}}
  = (\ol R_{11}+  |\chi|^2 )\ol A_{01}
  -(d-2)\nu_0\ol B_{11} + f(\ol A_{ij})
\label{ODE_system_A_1}
\end{eqnarray}
\underline{$(\mu\nu)=(1A)$:}
\begin{eqnarray}
 \lefteqn{(\partial_1+\frac{d-4}{2(d-2)}\tau -\nu^0\partial_1\nu_0)\ol{\nabla_0A_{1A}}
 - \sigma_A{}^B\ol{\nabla_0A_{1B}}
= \frac{1}{2}(\ol R_{11}+ |\chi|^2)\ol A_{0A}}
\nonumber
\\
 &&+ (\ol R_{1A1}{}^{B} + \chi_C{}^B\chi_{A}{}^C)\ol A_{0B}
  -(d-2)\nu_0\ol B_{1A} + f_A(\ol A_{ij}, \ol A_{01}, \ol{\nabla_0 A_{11}}) \hspace{1em}
\label{ODE_system_A_2}
\end{eqnarray}
\underline{$(\mu\nu)=(AB)$:}
\begin{eqnarray}
  (\partial_1 + \frac{d-6}{2(d-2)}\tau-\ol \Gamma^0_{01})\ol{\nabla_0A_{AB}}
 -2\sigma_{(A}{}^C\ol {\nabla_0A_{B)C}} \hspace{10em}
\nonumber
\\
=  -\nu^0(\ol R_{1A1B} + \chi_{AC}\chi_{B}{}^C)\ol A_{00}
  -(d-2)\nu_0\ol B_{AB}
+ f_{AB}(\ol A_{ij}, \ol A_{0i}, \ol{\nabla_0 A_{1i}})
\label{ODE_system_A_3}
\end{eqnarray}
\underline{$(\mu\nu)=(01)$:}
\begin{eqnarray}
(\partial_1 + \frac{\tau}{2}+ \ol\Gamma^1_{11}-2\nu^0\partial_1\nu_0)\ol{\nabla_0A_{01}}
  =  -(d-2)\nu_0\ol B_{01}
+ f(\ol A_{\mu\nu}, \ol{\nabla_0 A_{ij}})
\label{ODE_system_A_4}
\end{eqnarray}
\underline{$(\mu\nu)=(0A)$:}
\begin{eqnarray}
(\partial_1 +\frac{d-4}{2(d-2)}\tau - 2\ol\Gamma^0_{01})\ol{\nabla_0A_{0A}}
-\sigma_A{}^B\ol{\nabla_0A_{0B}}
\nonumber
\\
  = -(d-2)\nu_0\ol B_{0A}
+ f_A(\ol A_{\mu\nu}, \ol{\nabla_0 A_{ij}}, \ol{\nabla_0 A_{01}})
\label{ODE_system_A_5}
\end{eqnarray}
\underline{$(\mu\nu)=(00)$:}
\begin{eqnarray}
(\partial_1 + \frac{\tau}{2} -3 \ol \Gamma^0_{01})\ol{\nabla_0A_{00}}
= -(d-2)\nu_0\ol B_{00}
+ f(\ol A_{\mu\nu}, \ol{\nabla_0 A_{ij}}, \ol{\nabla_0 A_{0i}})
\label{ODE_system_A_6}
\end{eqnarray}

\subsubsection{An auxiliary lemma}

\begin{lemma}
\label{vanishing_nabla_A}
Assume that the wave equations for $X$ and $Y$, \eq{wave_X0} and \eq{wave_Y0}, are fulfilled.
Assume further that $\ol A_{\mu\nu} =0 = \ol B_{\mu\nu}$
 on either a light-cone or two transversally intersecting null hypersurfaces.
Then $\ol{\nabla_{\sigma} A_{\mu\nu}}=0$.
\end{lemma}
\begin{proof}
We start with the light-cone case.
By assumption the equations \eq{ODE_system_A_1}-\eq{ODE_system_A_6} hold.
Invoking  $\ol A_{\mu\nu} =0 = \ol B_{\mu\nu}$ they become
%
%\begin{eqnarray*}
% \ol{\Box_g A_{11}} = 0 & \Longrightarrow  &  (\partial_1 + \frac{1}{2}\tau - \ol \Gamma^1_{11}-\nu^0\partial_1\nu_0)\ol{\nabla_0A_{11}} =0
% \;,
%\\
% \ol{\Box_g A_{1A}} = 0 & \Longrightarrow  &   (\partial_1+\frac{n-3}{2(n-1)}\tau -\nu^0\partial_1\nu_0)\ol{\nabla_0A_{1A}}
% - \sigma_A{}^B\ol{\nabla_0A_{1B}}
% \\
% &&= \ol \Gamma^1_{1A}\ol{\nabla_0A_{11}}
% \;,
%\\
% \ol{\Box_g A_{AB}} = 0 & \Longrightarrow  &   (\partial_1 - \frac{n-5}{2(n-1)}\tau+\ol \Gamma^1_{11}-\nu^0\partial_1\nu_0)\ol{\nabla_0A_{AB}}
% -2\sigma_{(A}{}^C\ol {\nabla_0A_{B)C}}
%\\
%&& =2\ol\Gamma^1_{1(A}\ol {\nabla_{|0|}A_{B)1}}
% \;,
%\\
% \ol{\Box_g A_{01}}= 0 & \Longrightarrow  &   (\partial_1 + \frac{1}{2}\tau+ \ol\Gamma^1_{11}-2\nu^0\partial_1\nu_0)\ol{\nabla_0A_{01}}
%\\
%&&= \ol \Gamma^1_{01}\ol{\nabla_0A_{11}} + \ol\Gamma^B_{01}\ol{\nabla_0A_{1B}}
% \;,
%\\
% \ol{\Box_g A_{0A}} = 0 & \Longrightarrow  &    (\partial_1 +\frac{n-3}{2(n-1)}\tau + 2\ol\Gamma^1_{11}- 2\nu^0\partial_1\nu_0)\ol{\nabla_0A_{0A}}
%-\sigma_A{}^B\ol{\nabla_0A_{0B}}
%\\
% && = \ol \Gamma^1_{01}\ol{\nabla_0 A_{1A}} + \ol\Gamma^1_{1A}\ol{\nabla_0A_{01}} + \ol \Gamma^B_{01}\ol{\nabla_0A_{AB}}
% \;,
%\\
% \ol{\Box_g A_{00}} = 0 & \Longrightarrow  &   (\partial_1 + \frac{1}{2}\tau +3 \ol \Gamma^1_{11}-3\nu^0\partial_1\nu_0)\ol{\nabla_0A_{00}}
%\\
%&&= 2\ol \Gamma^1_{01}\ol{\nabla_0A_{01}} +2\ol\Gamma^B_{01}\ol{\nabla_0A_{0B}}
% \;,
%\end{eqnarray*}
%
\begin{eqnarray}
 \label{31VIII13.1}
 (\partial_1 + \frac{\tau}{2}- \ol\Gamma^0_{01} - 2\ol \Gamma^1_{11})\ol{\nabla_0A_{11}}  &=& 0
 \;,
\\
  (\partial_1+\frac{d-4}{2(d-2)}\tau -\nu^0\partial_1\nu_0)\ol{\nabla_0A_{1A}}
 - \sigma_A{}^B\ol{\nabla_0A_{1B}}
 &=& f_A(\ol{\nabla_0A_{11}})
 \;,
\\
    (\partial_1 + \frac{d-6}{2(d-2)}\tau-\ol \Gamma^0_{01})\ol{\nabla_0A_{AB}}
 -2\sigma_{(A}{}^C\ol {\nabla_0A_{B)C}}
 &=& f_{AB}(\ol{\nabla_0A_{1i}})
 \;,
\\
  (\partial_1 + \frac{\tau}{2}+ \ol\Gamma^1_{11}-2\nu^0\partial_1\nu_0)\ol{\nabla_0A_{01}}
 &=& f(\ol{\nabla_0A_{ij}})
 \;,
\\
   (\partial_1 +\frac{d-4}{2(d-2)}\tau - 2\ol\Gamma^0_{01})\ol{\nabla_0A_{0A}}
-\sigma_A{}^B\ol{\nabla_0A_{0B}}
 &=& f_A(\ol{\nabla_0A_{ij}}, \ol{\nabla_0A_{01}})
 \;,
  \phantom{xxxx}
\\
   (\partial_1 + \frac{\tau}{2} -3 \ol \Gamma^0_{01})\ol{\nabla_0A_{00}}
 &=& f(\ol{\nabla_0A_{ij}}, \ol{\nabla_0A_{0i}})
 \;.
  \label{31VIII13.2}
\end{eqnarray}
This is a hierarchical system of  Fuchsian ODEs which can be solved step-by-step.
Regularity requires the following behaviour of $\ol{\nabla_0 A_{\mu\nu}}$ near the vertex:
\begin{eqnarray}
&  \ol{\nabla_0 A_{11}} = O(1)
\;, \quad
  \ol{\nabla_0 A_{1A}} = O(r)
\;, \quad
  \ol{\nabla_0 A_{AB}} = O(r^2)
\;,&
\\
&  \ol{\nabla_0 A_{01}} = O(1)
\;, \quad
  \ol{\nabla_0 A_{0A}} = O(r)
\;, \quad
  \ol{\nabla_0 A_{00}} = O(1)
\;.&
 \label{31VIII13.3}
\end{eqnarray}
Taking the behaviour of the metric components at the tip of the cone into account, cf.\ the formulae (4.41)-(4.51) in \cite{CCM2},
which hold in any sufficiently regular gauge,
 \eq{31VIII13.1}-\eq{31VIII13.3} and standard results on Fuchsian ODEs (cf.\ e.g.\ \cite[Appendix A]{CP1}) imply
$\ol{\nabla_0 A_{\mu\nu}}=0$.

In the case of two transversally intersecting null hypersurfaces one can derive the same hierarchical system of ODEs on $N_1$ and $N_2$, respectively, which now is a system of regular ODEs.
The assumption $\ol A_{\mu\nu}=0$ implies $\nabla_{\sigma} A_{\mu\nu}|_S=0$. We thus have vanishing initial data for the ODEs and
the unique solutions are $\ol{\nabla_1 A_{\mu\nu}}|_{N_1}=0$ and $\ol{\nabla_2 A_{\mu\nu}}|_{N_2}=0$.
\qed
\end{proof}

\subsection{A special case: $\Theta =1$}
\label{sec_special_case}

Let us briefly analyse the implications of Theorem~\ref{KID_eqns_main} in the special case where the conformal factor $\Theta$ is identical to one,
\begin{equation*}
 \Theta =1
\end{equation*}
 (note that thereby the gauge freedom
to prescribe the Ricci scalar is lost).
 Then the unphysical spacetime can be identified with the physical spacetime.
The conformal field equations \eq{conf1}-\eq{conf6} imply the equations
\begin{eqnarray*}
 &s= \frac{1}{2(d-1)} \lambda\;,&
 \\
& L_{\mu\nu} = s g_{\mu\nu} \quad \Longleftrightarrow \quad  R_{\mu\nu} = \lambda g_{\mu\nu}
 \;,
\end{eqnarray*}
i.e.\ in particular the vacuum Einstein equations hold.

Let us analyse the conditions (i)-(vi) of Theorem~\ref{KID_eqns_main} in this setting:
Condition (iii) gives $\ol Y=0$, which we take as initial data for the wave equation (ii) which then implies $Y=0$,
i.e.\ $X$ needs to be divergence-free, as desired.
We observe that (iv) is automatically satisfied. Moreover,
\begin{equation*}
  B_{\mu\nu} = \mcL_X L_{\mu\nu}  = s\mcL_X g_{\mu\nu} = 2s\nabla_{(\mu}X_{\nu)}
   \;,
\end{equation*}
so (vi) follows from (v). To sum it up, the hypotheses of Theorem~\ref{KID_eqns_main} are satisfied if and only if there
is a vector field $X$ which satisfies
\begin{eqnarray*}
 \Box_g X_{\mu} + \lambda X_{\mu} &=& 0
 \;,
\\
 \ol{\nabla_{(\mu}X_{\nu)}}&=& 0
 \;.
\end{eqnarray*}
This was the starting point of the analysis in \cite{CP1}.

\section{KID equations in four dimensions}
\label{KID_4dim}

Theorem~\ref{KID_eqns_main} can be applied to dimensions $d\geq 5$ only when the conformal factor $\Theta$ is bounded away from zero.
In fact, this situation is rather uninteresting since then there is no need to pass to a conformally rescaled spacetime
(or to put it differently, it is just a matter of gauge to set $\Theta=1$).
One reason why we included this case, though, was to emphasize that their arise difficulties when one tries to go from four to higher dimensions (which is in line with the observation that the conformal field equations provide a good evolution
system only in four spacetime dimensions).
Another reason was to be able to consider the limiting case $\Theta=1$ in any dimension $d\geq 4$ where the unphysical spacetime can be identified with the
physical spacetime, and to compare the resulting equations with those in \cite{CP1}.
This is also a reason why we avoided to make the common gauge choice $R=0$:
$\Theta=1$ is compatible with $R=0$ solely when the cosmological constant vanishes.
Henceforth we restrict attention to $d=4$ spacetime dimensions.

\subsection{A stronger version  of Theorem~\ref{KID_eqns_main} for light-cones}

A more careful analysis  of the computations made in the proof of Lemma~\ref{vanishing_nabla_A}
will lead us to a refinement of Theorem~\ref{KID_eqns_main}. We first treat the light-cone case.
We will assume the vanishing of $\ol A_{ij}$ and $\ol B_{ij}$ as well as the validity of the wave equations  \eq{wave_X0} and \eq{wave_Y0} for $X$ and $Y$,  and explore the consequences concerning the vanishing of other components
of these tensors, including certain transverse derivatives thereof.

\subsubsection{Vanishing of $\ol A_{0\mu}$}

We start with the identity
\begin{eqnarray}
   \nabla_{\nu}A_{\mu}{}^{\nu}
   -\frac{1}{2}\nabla_{\mu}A_{\nu}{}^{\nu}
\,\equiv \,
 \Box_g  X_{\mu}  + R_{\mu}{}^{\nu}X_{\nu}+ 2\nabla_{\mu} Y
   \;,
   \label{waveX_divergenceA}
\end{eqnarray}
which together with the wave equation \eq{wave_X0} for $X$ implies
\begin{eqnarray}
   \nabla_{\nu}A_{\mu}{}^{\nu}
   -\frac{1}{2}\nabla_{\mu}A_{\nu}{}^{\nu}
\,=\,0
   \;.
   \label{waveX_divergenceA2}
\end{eqnarray}
On the initial surface that yields in adapted null coordinates,
\begin{eqnarray}
%&&\hspace{-3em} \ol{ \Box_g  X_{\mu} } +\ol  R_{\mu}{}^{\nu}\ol X_{\nu}+2\ol{\nabla_{\mu} Y}
%\nonumber
%\\
0&=&
  \nu^0( 2\ol{\nabla_{(0}A_{1)\mu }} - \ol{\nabla_{\mu}A_{01}})
  +  \ol g^{11}( \nabla_{1}\ol A_{\mu 1} -\frac{1}{2}\ol{\nabla_{\mu}A_{11}} )
  \nonumber
  \\
  && + \ol g^{1A}(2 \nabla_{(1}\ol A_{A)\mu}- \ol{\nabla_{\mu}A_{1A}})
  +  \ol g^{AB} (\nabla_{A}\ol A_{\mu B}  -\frac{1}{2} \ol{\nabla_{\mu}A_{AB}})
   \;.
   \label{wave_X_comp}
\end{eqnarray}
In the following we shall always assume that \eq{wave_X0} and \eq{wave_Y0} hold (and thereby in particular \eq{wave_X_comp} and
\eq{ODE_system_A_1}-\eq{ODE_system_A_3}).

With the assumptions $\ol A_{ij}=0$ and $\ol B_{11}=0$
equation \eq{ODE_system_A_1} becomes
\begin{equation*}
  (\partial_1+ \frac{1}{2}\tau - \ol\Gamma^1_{11} -\nu^0\partial_1\nu_0)\ol{ \nabla_0 A_{11}  }
 \,=\, -\ol A_{01} ( \partial_1 - \ol \Gamma^1_{11} )\tau
 \;,
\end{equation*}
where we have fallen back on the identity \cite{CCM2}
\begin{equation}
\ol R_{11}\equiv -(\partial_1 - \ol \Gamma^1_{11})\tau - |\chi|^2
\;.
\label{id_R11}
\end{equation}
Moreover, we deduce from the $\mu=1$-component of \eq{wave_X_comp} that
\begin{equation}
%\,=\,  \ol{  \Box_g  X_{1}}  +  \ol R_{1}{}^{\nu} \ol X_{\nu}+\frac{1}{2}\nabla_{1} \ol Y
 \tau \ol A_{01} + \ol{\nabla_0A_{11}} \,=\, 0
%\nu^0 (\tau \ol A_{01} + \ol{\nabla_0A_{11}})
\;,
\label{relation_for_A01}
\end{equation}
which leads us to an ODE satisfied by $\ol A_{01}$,
\begin{equation}
  \tau ( \partial_1 + \frac{1}{2}\tau-\nu^0\partial_1\nu_0) \ol A_{01}
 \,=\,   0
 \;.
 \label{ODE_A01}
\end{equation}
 Since $\tau$ has no zeros near the vertex it follows from regularity (which requires $\ol A_{01}$ to be bounded near the vertex) that there $\ol A_{01}=0$ (and thus $\ol{\nabla_0A_{11}}=0$).
Even more, $\ol A_{01}$ will automatically vanish on the closure of those sets on which $\tau$ is non-zero.

Next we assume   $\ol A_{ij}=0$, $\ol A_{01}=0$, $\ol{\nabla_0A_{11}}=0$, $\ol B_{1A}=0$.
Then, due to \eq{ODE_system_A_2}, \eq{id_R11} and the identity
\begin{eqnarray}
 \ol R_{1A1}{}^B &\equiv& -(\partial_1- \ol \Gamma^{1}_{11})\chi_{A}{}^B - \chi_A{}^C\chi_C{}^B
\label{ident_Riem_comp}
 \end{eqnarray}
we have
\begin{eqnarray*}
    (\partial_1 -\nu^0\partial_1\nu_0 )\ol{\nabla_{0}A_{1A}}
 -\sigma_A{}^B\ol{\nabla_0A_{1B}}
=
    - \ol A_{0A}(\partial_1 - \ol \Gamma^1_{11} )\tau
     -\ol A_{0B}(\partial_1- \ol \Gamma^{1}_{11} ) \sigma_A{}^B
    \;.
\end{eqnarray*}
With the current assumptions the $\mu=A$-component of \eq{wave_X_comp} can be written as
\begin{eqnarray}
%  0
%  &=&  \ol{\Box_g  X_{A}}  + \ol R_{A}{}^{\nu}\ol X_{\nu}+2\nabla_{A} \ol Y
% \\
% &=&
  \ol{\nabla_{0}A_{1A }} + ( \partial_1 +  \tau - \ol\Gamma^0_{01}  )  \ol A_{0A}
   &=& 0
 \;,
\label{rln_nabla0_A1A}
\end{eqnarray}
whence
\begin{eqnarray}
  (\partial_1 +\tau -\nu^0\partial_1\nu_0 )[( \partial_1 -\ol\Gamma^0_{01} )  \ol A_{0A}
     -  \sigma_A{}^B \ol A_{0B}] =0
     \label{ODE_A0A}
  \;.
\end{eqnarray}
Regularity requires $\ol A_{0A}=O(r)$, and $\ol A_{0A}=0$ is the only solution with this property.
Then, of course, $\ol{\nabla_{0}A_{1A }}$ vanishes as well.

In the final step we assume
(in addition to the validity of \eq{wave_X0} and \eq{wave_Y0})
$\ol A_{ij}=0$, $\ol A_{0i}=0$, $\ol{\nabla_0A_{1i}}=0$ and $\ol g^{AB}\ol B_{AB}=0$.
Taking the $\ol g_{AB}$-trace of \eq{ODE_system_A_3} and using again \eq{id_R11} we find
\begin{eqnarray*}
 ( \partial_1 +\frac{1}{2}\tau+\ol\Gamma^1_{11}- \nu^0\partial_1\nu_0)(\ol{g^{AB}\nabla_0 A_{AB}})
 &=& \nu^0\ol A_{00} (\partial_1 - \ol \Gamma^1_{11})\tau
  \;.
\end{eqnarray*}
From the $\mu=0$-component of \eq{wave_X_comp} we derive the equation
\begin{eqnarray*}
 %0
% &=& \ol{ \Box_g  X_{0} } +\ol  R_{0}{}^{\nu}\ol X_{\nu}+\frac{1}{2}\ol{\nabla_{0} Y}
% \\
% &=&
  \nu^0(   \partial_1+\tau +2\ol\Gamma^1_{11}-2\nu^0\partial_1\nu_0)\ol A_{00}
  -\frac{1}{2} \ol{ g^{AB}\nabla_{0}A_{AB}} &=& 0
   \;,
\end{eqnarray*}
and end up with an ODE satisfied by $\ol A_{00}$,
\begin{eqnarray}
 ( \partial_1 +\tau+\ol\Gamma^1_{11}- \nu^0\partial_1\nu_0)[ \nu^0(   \partial_1+\frac{1}{2}\tau +2\ol\Gamma^1_{11}-2\nu^0\partial_1\nu_0)\ol A_{00} ] &=& 0
  \;.
  \label{ODE_A00}
\end{eqnarray}
Regularity requires $\ol A_{00} =O(1)$, which leads to $\ol A_{00}=0$,
which in turn implies $\ol{ g^{AB}\nabla_{0}A_{AB}}=0$.
Assuming $\ol B_{AB}=0$ we further have $\ol{\nabla_0 A_{AB}}=0$ due to \eq{ODE_system_A_3}.

Altogether we have proved the lemma
\begin{lemma}
\label{lemma_A0mu}
 Assume that \eq{wave_X0} and \eq{wave_Y0} hold, and that $\ol A_{ij} = \ol A_{01}= 0 = \ol B_{ij}$.
Then $\ol A_{0\mu} = 0 $ and $\ol{\nabla_0 A_{ij}}=0 $.
On the closure of those sets where $\tau$ is non-zero, in particular sufficiently close to the vertex of the cone, the assumption
$\ol A_{01}= 0$ is \textit{not} needed, but follows from the remaining hypotheses.
\end{lemma}

\subsubsection{Vanishing of $\ol B_{0\mu}$}

By the second Bianchi identity we have
\begin{eqnarray*}
 \lefteqn{\nabla_{\nu} B_{\mu}{}^{\nu}- \frac{1}{2}\nabla_{\mu}B_{\nu}{}^{\nu} \equiv
  A_{\alpha\beta}(\nabla^{\alpha}L_{\mu}{}^{\beta} - \frac{1}{2}\nabla_{\mu}L^{\alpha\beta})}
\\
&&+L_{\mu}{}^{\kappa}(\Box_g X_{\kappa} + R_{\kappa}{}^{\alpha}X_{\alpha}+ 2\nabla_{\kappa}Y)
  + \frac{1}{2} \nabla_{\mu}(\Box_g Y+ \frac{1}{6}X^{\nu}\nabla_{\nu}R+ \frac{1}{3} R Y)
   \;.
\end{eqnarray*}
Assuming the wave equations \eq{wave_X0} and \eq{wave_Y0} for $X$ and $Y$ as well as $\ol A_{\mu\nu}=0$ this induces
 on the initial surface the relation
\begin{eqnarray}
 \ol{\nabla_{\nu} B_{\mu}{}^{\nu}} - \frac{1}{2}\ol{\nabla_{\mu}B_{\nu}{}^{\nu}} &= & 0
 \;.
\end{eqnarray}
As for $A_{\mu\nu}$, equation \eq{wave_X_comp}, we deduce that in adapted coordinates we have
\begin{eqnarray}
%&&\hspace{-3em} \ol{ \Box_g  X_{\mu} } +\ol  R_{\mu}{}^{\nu}\ol X_{\nu}+2\ol{\nabla_{\mu} Y}
%\nonumber
%\\
0&=&
  \nu^0( 2\ol{\nabla_{(0}B_{1)\mu }} - \ol{\nabla_{\mu}B_{01}})
  +  \ol g^{11}( \nabla_{1}\ol B_{\mu 1} -\frac{1}{2}\ol{\nabla_{\mu}B_{11}} )
  \nonumber
  \\
  && + \ol g^{1A}(2 \nabla_{(1}\ol B_{A)\mu}- \ol{\nabla_{\mu}B_{1A}})
  +  \ol g^{AB} (\nabla_{A}\ol B_{\mu B}  -\frac{1}{2} \ol{\nabla_{\mu}B_{AB}})
   \;.
   \label{wave_X_comp2}
\end{eqnarray}
Recall that \eq{wave_X0}, \eq{wave_Y0} and the conformal field equations imply a
 wave equation \eq{waveeqn_B} which is satisfied by $B_{\mu\nu}$.
Assuming $\ol{A_{\mu\nu}}=0$, $\ol{\nabla_0 A_{ij}}=0$ and $\ol B_{ij}=0$, evaluation on the initial surface yields
\begin{eqnarray}
\ol{\Box_g B_{ij}} &=&
 2(\ol g_{ij}\ol L^{\alpha \beta}-   \ol R_{i}{}^{\alpha}{}_{j}{}^{\beta})\ol B_{\alpha\beta}
 - 2 \ol R_{(i}{}^{\kappa} \ol B_{j)\kappa}
\nonumber
\\
&& + 2\ol L^{\alpha\beta}(\ol {\nabla_{\beta}\nabla_{[\alpha}A_{j]i}}  -  \ol {\nabla_{i}\nabla_{[\alpha}A_{j]\beta}})
\label{wave_B_cone}
\;.
\end{eqnarray}
In adapted null coordinates we have, as for the corresponding expression \eq{identity_wave_Agen} for $A_{\mu\nu}$,
\begin{eqnarray}
 \ol{\Box_g B_{ij}} &=& 2\nu^0(\ol{\nabla_1\nabla_0B_{ij}} + \ol R_{01(i}{}^{\alpha}\ol  B_{j)\alpha})
  + \ol g^{11} \ol{\nabla_1\nabla_1B_{ij}}
   \nonumber
  \\
&&  + 2\ol g^{1A} (\ol{\nabla_1\nabla_AB_{ij}}
 + \ol R_{A1(i}{}^{\alpha} \ol B_{j)\alpha})
  + \ol g^{AB}\ol{\nabla_A\nabla_BB_{ij}}
  \;.
  \label{identity_wave_B}
\end{eqnarray}
Moreover, we have seen that \eq{wave_X0} and \eq{wave_Y0} imply the wave equation \eq{waveeqn_nabla_A}
satisfied by $\nabla_{\sigma}A_{\mu\nu}$.
Assuming $\ol{A_{\mu\nu}}=0$ and $\ol{\nabla_0 A_{ij}}=0$ we compute its trace on the initial
\begin{eqnarray}
\ol{\Box_g\nabla_{0}A_{ij}}  &=&
  2\nu^0 \ol R_{1(i}\ol{\nabla_{0}A_{j)0}}
 - 2\ol R_{i}{}^{\alpha}{}_{j}{}^{\kappa}\ol{\nabla_{0}A_{\alpha\kappa} }
\nonumber
\\
 &&
  +4(\nu^0)^2 \ol R_{011(i}\ol{\nabla_{|0|}A_{j)0}}
  - 4\ol{\nabla_{0}B_{ij}}
\;.
\label{wave_nabla_A_cone}
\end{eqnarray}
In adapted null coordinates and with the current assumptions the left-hand side becomes
\begin{eqnarray}
\ol{\Box_g\nabla_{0}A_{ij}}  &=&
2\nu^0 \ol{\nabla_1\nabla_0 \nabla_{0}A_{ij}} +
2\nu^0\ol R_{01(i}{}^{\mu} \ol{\nabla_{|0|}A_{j)\mu}}
 +2\ol g^{1A}\ol{ \nabla_1\nabla_A \nabla_{0}A_{ij}}
\nonumber
\\
&&
+2\ol g^{1A}\ol  R_{A1(i}{}^{\mu} \ol{\nabla_{|0|}A_{j)\mu}}
+\ol g^{AB} \ol{\nabla_A\nabla_B \nabla_{0}A_{ij}}
 \;.
\label{identity_wave_nabla_A}
\end{eqnarray}

Recall that $\ol A_{\mu\nu}=0$ suffices to establish $\ol B_{\alpha}{}^{\alpha}=0$. In that case  $\ol B_{ij}=0$ implies
\begin{equation}
 \ol B_{01} \,=\, 0
\;,
\end{equation}
and, by \eq{ODE_system_A_4},
\begin{eqnarray*}
 \ol{\nabla_0 A_{01}} \,=\, 0
\;.
\end{eqnarray*}
As for $A_{\mu\nu}$, the $\mu=1$-component of \eq{wave_X_comp2} yields
\begin{equation}
 \tau \ol B_{01} + \ol{\nabla_0 B_{11}} \,=\,0 \quad \Longrightarrow \quad  \ol{\nabla_0 B_{11}} \,=\,0
\;.
\end{equation}
The $(ij)=(11)$-component of \eq{wave_nabla_A_cone} reads
\begin{eqnarray*}
\ol{\Box_g\nabla_{0}A_{11}}  \,=\, 0 & \overset{\eq{identity_wave_nabla_A}}{\Longrightarrow }& (\partial_1 + \frac{1}{2}\tau -2\nu^0\partial_1\nu_0)\ol{\nabla_0\nabla_0 A_{11}} =0
\\
& \Longrightarrow &  \ol{\nabla_0\nabla_0 A_{11}} =0
\end{eqnarray*}
by regularity.

At this stage we can and will assume $\ol A_{\mu\nu} = \ol{\nabla_0 A_{ij}} = \ol{\nabla_0 A_{01}} =\ol{\nabla_0\nabla_0 A_{11}}
= \ol B_{ij} = \ol B_{01} = \ol{\nabla_0 B_{11}}=0$.
Then, with \eq{ODE_system_A_5}, we find for the $(ij)=(1A)$-components of \eq{wave_B_cone}
\begin{eqnarray*}
\ol{\Box_g B_{1A}} &=&
 2  \nu^0 \ol R_{1A1}{}^{B}\ol B_{0B}
+ \frac{1}{2}(\nu^0)^2\ol R_{11}\ol {\nabla_{0}\nabla_{0}A_{1A}}
\\
&& +  \frac{1}{2}\tau (\nu^0)^2\ol R_{11}\ol{\nabla_{0}A_{0A}}
+ \frac{1}{2}(\nu^0)^2\ol R_{11}\sigma_A{}^B \ol{\nabla_{0}A_{0 B}}
\;.
\end{eqnarray*}
The $(ij)=(1A)$-components of \eq{identity_wave_B} read
\begin{eqnarray*}
    \ol{\Box_g B_{1A}}
    &=&
    2\nu^0(\partial_1 -\nu^0\partial_1\nu_0 )\ol{\nabla_{0}B_{1A}}
    -  2\nu^0\sigma_A{}^B\ol{\nabla_0B_{1B}}
    \\
    &&
    - \nu^0(  |\chi|^2\ol B_{ 0A}
     + 2\chi_{A}{}^C\chi_C{}^B\ol B_{0B})
      \;.
\end{eqnarray*}
Equating both expressions for $ \ol{\Box_g B_{1A}}$  and using \eq{ident_Riem_comp} we deduce that
\begin{eqnarray}
 &&\hspace{-3em}   2(\partial_1 -\nu^0\partial_1\nu_0 )\ol{\nabla_{0}B_{1A}}
    -  2\sigma_A{}^B\ol{\nabla_0B_{1B}}
 \nonumber
\\
 &=& -2 \ol B_{0B} (\partial_1-\ol \Gamma^1_{11})\chi_A{}^B
+ |\chi|^2\ol B_{ 0A}
+ \frac{1}{2}\nu^0\ol R_{11}\ol {\nabla_{0}\nabla_{0}A_{1A}}
\nonumber
\\
&& +  \frac{1}{2}\tau \nu^0\ol R_{11}\ol{\nabla_{0}A_{0A}}
+ \frac{1}{2}\nu^0\ol R_{11}\sigma_A{}^B \ol{\nabla_{0}A_{0 B}}
      \;.
\label{Fuchsian_A1}
\end{eqnarray}
Evaluation of \eq{wave_nabla_A_cone} for $(ij)=(1A)$ leads to
\begin{eqnarray*}
\ol{\Box_g\nabla_{0}A_{1A}}  &=&
  \nu^0 \ol R_{11}\ol{\nabla_{0}A_{0A}}
 + 2\nu^0\ol R_{1A1}{}^{B}\ol{\nabla_{0}A_{0B} }
  - 4\ol{\nabla_{0}B_{1A}}
 \;,
\end{eqnarray*}
while \eq{identity_wave_nabla_A} becomes
\begin{eqnarray*}
\ol{\Box_g\nabla_{0}A_{1A}}  &=&
2 \nu^0(\partial_1  + \ol \Gamma^{1}_{11}-2\nu^0\partial_1\nu_0 ) \ol{\nabla_0 \nabla_{0}A_{1 A}}  - 2\nu^0\sigma_A{}^B \ol{\nabla_0 \nabla_{0}A_{1B}}
\\
&&
 -\nu^0|\chi|^2\ol{\nabla_{0}A_{0 A}}
 -2\nu^0\chi_{A}{}^C\chi_C{}^B \ol{\nabla_{0}A_{0B}}
 \;.
\end{eqnarray*}
Using \eq{id_R11} and \eq{ident_Riem_comp} we end up with
\begin{eqnarray}
&&\hspace{-2em}2 (\partial_1  +\ol \Gamma^{1}_{11}-2\nu^0\partial_1\nu_0 )\ol{ \nabla_0 \nabla_{0}A_{1 A}}  - 2\sigma_A{}^B \ol{\nabla_0 \nabla_{0}A_{1B}}
\nonumber
\\
&=&
 -2 \ol{\nabla_{0}A_{0A}}(\partial_1-\ol\Gamma^1_{11})\tau
 - 2\ol{\nabla_{0}A_{0B} }(\partial_1-\ol\Gamma^1_{11})\sigma_A{}^B
  - 4\nu_0\ol{\nabla_{0}B_{1A}}
 \;. \hspace{1em}
\label{Fuchsian_A2}
\end{eqnarray}
The $\mu=A$-components of  \eq{wave_X_comp2} give, again in close analogy to
 the corresponding equations for $A_{\mu\nu}$,
\begin{eqnarray}
 ( \partial_1 +  \tau -\ol\Gamma^0_{01}  )  \ol B_{0A}+    \ol{\nabla_{0}B_{1A }}
   &=& 0
   \;.
\label{Fuchsian_A3}
\end{eqnarray}
Recall that by \eq{ODE_system_A_5} we have
\begin{eqnarray}
 (\partial_1-2\ol\Gamma^0_{01})\ol {\nabla_0 A_{0A}} - \sigma_A{}^B\ol{\nabla_0A_{0B}} = -2\nu_0 \ol B_{0A}
 \;.
\label{Fuchsian_A4}
\end{eqnarray}
Taking the behaviour of the metric components at the vertex into account, cf.\ \cite[Section~4.5]{CCM2},
%which holds in any sufficiently regular gauge,
we observe that the ODE-system \eq{Fuchsian_A1}-\eq{Fuchsian_A4}
for $\ol B_{0A}$, $\ol{\nabla_0 B_{1A}}$, $\ol{\nabla_0 A_{0A}}$ and $\ol{\nabla_0\nabla_0 A_{1A}}$
is of the form
\begin{eqnarray*}
\left[ \partial_1
  +
\begin{pmatrix}   2r^{-1} +O(r)&1 &0 &0 \\ -2r^{-2}+ O(1) &O(r) &O(r^{-1}) &O(1) \\ 2+O(r^2)& 0 &O(r) &0\\ 0&2+O(r^2)&  -2r^{-2}+O(1)&O(r)     \end{pmatrix}
\right]
 \begin{pmatrix} \ol B_{0A} \\ \ol{\nabla_{0}B_{1A}} \\ \ol {\nabla_0 A_{0A}}  \\ \ol{ \nabla_0 \nabla_{0}A_{1 A}}
\end{pmatrix} =0
 \;,
\end{eqnarray*}
where each matrix entry is actually a $2\times 2$-matrix.
Regularity requires
\begin{equation*}
  \ol B_{0A} \;,\,  \ol{\nabla_{0}B_{1A}} \;, \,   \ol {\nabla_0 A_{0A}}  \;, \,   \ol{ \nabla_0 \nabla_{0}A_{1 A}} = O(r)
\;.
\end{equation*}
But then a necessary condition for \eq{Fuchsian_A3}
 to be satisfied is
\begin{equation}
  \ol B_{0A} =O(r^2)
\;.
\end{equation}
whence it follows from  \eq{Fuchsian_A4} and   \eq{Fuchsian_A2} that
\begin{equation}
   \ol {\nabla_0 A_{0A}} = O(r^3)\;, \quad \ol{ \nabla_0 \nabla_{0}A_{1 A}} =O(r^2)\;,
\end{equation}
Setting  $\tilde{\ol B}_{0A}:=r^{-2}\ol B_{0A}$, $ \tilde{ \ol{\nabla_{0}B_{1A}}}=r^{-1}\ol{\nabla_{0}B_{1A}}$,  $\tilde{\ol {\nabla_0 A_{0A}} }:=  r^{-3}\ol {\nabla_0 A_{0A}}$
and $  \tilde{\ol{ \nabla_0 \nabla_{0}A_{1 A}}} = r^{-2}\ol{ \nabla_0 \nabla_{0}A_{1 A}} $ the ODE-system adopts the form
\begin{eqnarray}
\left[ \partial_1
  +r^{-1} \begin{pmatrix}   4&1 &0 &0 \\ -2 &1 &0 &0 \\ 2& 0 & 3 &0\\ 0& 2&  -2& 2     \end{pmatrix}
+M \right] v=0
 \;,
\end{eqnarray}
where $M=O(r)$
%
%\begin{equation*}
% M = \begin{pmatrix}
% O(r) & 0 & 0 &0 \\ O(r) & O(r) & O(r) & O(r) \\ O(r) & 0 & O(r) & 0 \\ 0 & O(r) & O(r) & O(r)
%\end{pmatrix}
%\end{equation*}
%
is some matrix and
\begin{equation}
 v := \begin{pmatrix}  \tilde{\ol B}_{0A} \\ \tilde{\ol{\nabla_{0}B_{1A}}} \\  \tilde{\ol {\nabla_0 A_{0A}} }  \\ \tilde{\ol{ \nabla_0 \nabla_{0}A_{1 A}}}
\end{pmatrix}
=O(1)
\end{equation}
a smooth vector field. Setting
\begin{eqnarray}
 \tilde v:=T^{-1}v = O(1)
\;,
\label{smooth_tilde_v}
\end{eqnarray}
where
\begin{equation*}
 T:=  \begin{pmatrix} 0 & -1/2 &  -1/3 & 0 \\ 0 & 1/2 & 2/3 & 0 \\ -1 & -1/2&  2/3 & 0 \\ 2& 0 & 0 &1
\end{pmatrix}
\end{equation*}
is the change of basis matrix which transforms
the leading order matrix to Jordan normal form, we end up with the Fuchsian ODE-system
\begin{eqnarray}
 \partial_1\tilde v
  +r^{-1} \begin{pmatrix}   3&1 &0 &0 \\ 0 &3 &0 &0 \\ 0& 0 & 2 &0\\ 0& 0&  0& 2     \end{pmatrix}
\tilde v
+ \tilde M \tilde v=0\;, \quad \tilde M:=T^{-1}MT= O(r)
 \;.
\label{Fuchsian_system_1}
\end{eqnarray}
In Appendix~\ref{app_Fuchsian} it is shown that any solution of \eq{Fuchsian_system_1} which is $O(1)$ needs to vanish identically (take, in the notation used there, $\lambda=-1$).
%xxxxxxxxxxxxxxxxxxxxxxx
%
%Since the indicial matrix has a trivial kernel
%we deduce from \eq{Fuchsian_system_1} that the relation $\tilde v =O(r)$ must necessarily hold.
%Next we set
%%
%\begin{equation*}
% \tilde v^{(k)}:= r^{-k}\tilde v
%\;,
%\end{equation*}
%%
%which satisfies the equation
%%
%\begin{eqnarray}
% \partial_1\tilde v^{(k)}
%  +r^{-1} \begin{pmatrix}   3+k&1 &0 &0 \\ 0 &3+k &0 &0 \\ 0& 0 & 2+k &0\\ 0& 0&  0& 2+k     \end{pmatrix}
%\tilde v^{(k)}
%+ \tilde M \tilde v^{(k)}=0
% \;.
%\label{Fuchsian_system_1B}
%\end{eqnarray}
%%
%By induction we conclude that $\tilde v^{(k)}=O(r)$, i.e.
%%
%\begin{equation*}
% \tilde v = O(r^{\infty})
%\;.
%\end{equation*}
%%
%It is well-known
%\\ -- \\
%it seems that I might have mislead you concerning the issues; here the matrix you want to consider is the negative of the one you looked at, so the  largest eigenvalue -2, and hence every solution which is $o(r^{-2}$ vanishes. The calculation you just did is essentially the proof of this, together with a simple energy inequality.}
%that any solution of \eq{Fuchsian_system_1} which is $o(r^{k=3})$, $k$ being the largest eigenvalue of the indicial matrix, is identically zero.
%
%xxxxxxxxxxxxxx
Hence  $\ol B_{0A}=\ol{\nabla_0 B_{1A}}=\ol{\nabla_0 A_{0A}}=\ol{\nabla_0\nabla_0 A_{1A}}=0$, which we can and will
assume in the subsequent computations.

The $\ol g_{AB}$-trace of the $(ij)=(AB)$-component of \eq{wave_B_cone} reads
\begin{eqnarray*}
\ol g^{AB}\ol{\Box_g B_{AB}} &=&
\frac{1}{2}(\nu^0)^2\ol R_{11} \ol g^{AB}\ol {\nabla_{0}\nabla_{0}A_{AB}}
  - \frac{1}{2}\tau (\nu^0)^3 \ol R_{11} \ol{\nabla_{0}A_{00}}
 \;,
\end{eqnarray*}
for the corresponding component of \eq{identity_wave_B} we find
\begin{eqnarray*}
 \ol g^{AB} \ol{\Box_g B_{AB}} &\equiv&
2\nu^0  (\partial_1 + \frac{1}{2}\tau -\ol \Gamma^{0}_{01})(\ol g^{AB}\ol{ \nabla_0B_{AB}})
  + 2(\nu^0)^2|\chi|^2\ol B_{00}
  \;,
\end{eqnarray*}
and thus
\begin{eqnarray}
&&\hspace{-3em}   (\partial_1 + \frac{1}{2}\tau -\ol \Gamma^{0}_{01})(\ol g^{AB}\ol{ \nabla_0B_{AB}})
  + \nu^0|\chi|^2\ol B_{00}
\nonumber
\\
&=& \frac{1}{4}\nu^0 \ol R_{11}\ol g^{AB}\ol {\nabla_{0}\nabla_{0}A_{AB}}
  - \frac{1}{4}\tau (\nu^0)^2 \ol R_{11} \ol{\nabla_{0}A_{00}}
\label{Fuchsian_B1}
 \;.
\end{eqnarray}
From \eq{wave_nabla_A_cone} we deduce
\begin{eqnarray*}
\ol g^{AB}\ol{\Box_g\nabla_{0}A_{AB}}  &=&
 - 2(\nu^0)^2\ol R_{11}\ol{\nabla_{0}A_{00} }  - 4\ol g^{AB}\ol{\nabla_{0}B_{AB}}
\;,
\end{eqnarray*}
while
from \eq{identity_wave_nabla_A} we obtain
\begin{eqnarray*}
\ol g^{AB} \ol{\Box_g\nabla_{0}A_{AB}}  =
2\nu^0  (\partial_1+ \frac{1}{2}\tau -2\ol \Gamma^{0}_{01} )(\ol g^{AB}\ol{\nabla_{0} \nabla_{0}A_{AB}})
+ 2(\nu^0)^2|\chi|^2\ol{\nabla_{0}A_{0 0}}
 \;.
\end{eqnarray*}
Invoking \eq{id_R11} this leads us to the equation
\begin{eqnarray}
&&\hspace{-3em}  (\partial_1+ \frac{1}{2}\tau -2 \ol\Gamma^{0}_{01} )(\ol g^{AB}\ol{\nabla_{0} \nabla_{0}A_{AB}})
\nonumber
\\
 &=&  \nu^0 \ol{\nabla_{0}A_{00} } (\partial_1 - \ol \Gamma^1_{11})\tau  - 2\nu_0\ol g^{AB}\ol{\nabla_{0}B_{AB}}
 \;.
\label{Fuchsian_B2}
\end{eqnarray}
The $\mu=0$-component of \eq{wave_X_comp2} reads
\begin{eqnarray}
   (\partial_{1}+\tau  - 2\ol \Gamma^0_{01})\ol B_{00 }   -\frac{1}{2}\nu_0  \ol g^{AB} \ol{\nabla_{0}B_{AB}}
&=&0
   \;.
\label{Fuchsian_B3}
\end{eqnarray}
Recall that by \eq{ODE_system_A_6} we have
\begin{eqnarray}
(\partial_1 + \frac{1}{2}\tau -3 \ol \Gamma^0_{01})\ol{\nabla_0A_{00}}
+ 2\nu_0\ol B_{00} &= &0
\;.
\label{Fuchsian_B4}
\end{eqnarray}
Using again the results of \cite[Section~4.5]{CCM2} we find that the ODE-system \eq{Fuchsian_B1}-\eq{Fuchsian_B4}
for $\ol B_{00 } $, $\ol g^{AB}\ol{ \nabla_0B_{AB}}$, $\ol{\nabla_0A_{00}}$ and $\ol g^{AB}\ol{\nabla_{0} \nabla_{0}A_{AB}}$
is of the form
\begin{eqnarray*}
 \left[\partial_1
  +  \begin{pmatrix}   2r^{-1} +O(r)&-\frac{1}{2}+O(r^2) &0 &0 \\ 2r^{-2}+ O(1) &r^{-1} + O(r) &O(r^{-1}) &O(1) \\ 2+O(r^2)& 0 &r^{-1} + O(r) &0\\ 0&2+O(r^2)&  2r^{-2}+O(1)&r^{-1} + O(r)     \end{pmatrix}
\right]
 \begin{pmatrix} \ol B_{00 } \\ \ol g^{AB}\ol{ \nabla_0B_{AB}} \\ \ol {\nabla_0 A_{00}}  \\ \ol g^{AB}\ol{\nabla_{0} \nabla_{0}A_{AB}}
\end{pmatrix} =0
 \;.
\end{eqnarray*}
Due to regularity we have
\begin{equation*}
 \ol B_{00 } \:,\, \ol g^{AB}\ol{ \nabla_0B_{AB}} \;,\, \ol {\nabla_0 A_{00}}  \;,\, \ol g^{AB}\ol{\nabla_{0} \nabla_{0}A_{AB}}=O(1)
\;.
\end{equation*}
Even more, from \eq{Fuchsian_B3} we conclude that
\begin{equation}
\ol B_{00 } = O(r)
\;.
\end{equation}
From \eq{Fuchsian_B4} and  \eq{Fuchsian_B2} we  then deduce
\begin{equation}
 \ol {\nabla_0 A_{00}} = O(r^2)
\;, \quad
\ol g^{AB}\ol{\nabla_{0} \nabla_{0}A_{AB}}=O(r)
\;.
\end{equation}
In terms of  the rescaled fields $\tilde{\ol B}_{00}:=r^{-1}\ol B_{00}$,  $\tilde{\ol {\nabla_0 A_{00}} }:=  r^{-2}\ol {\nabla_0 A_{00}}$ and
$\tilde {\ol {g^{AB}\nabla_{0} \nabla_{0}A_{AB}}}= r^{-1}\ol g^{AB}\ol{\nabla_{0} \nabla_{0}A_{AB}}$  the ODE-system takes the form
%
%\begin{eqnarray*}
% \partial_1 \begin{pmatrix}\tilde{\ol B}_{00}\\ \ol g^{AB}\ol{ \nabla_0B_{AB}} \\ \tilde{\ol {\nabla_0 A_{00}} }  \\ \ol g^{AB}\ol{\nabla_{0} \nabla_{0}A_{AB}}
%\end{pmatrix}
%  +  \begin{pmatrix}   3r^{-1}+O(r)&-\frac{1}{2}r^{-1}+O(r) &0 &0 \\ 2r^{-1}+ O(r) &r^{-1} + O(r) &O(1) &O(1) \\ 2+O(r^2)& 0 &2r^{-1} + O(r) &0\\ %0&2+O(r^2)&  2r^{-1}+O(r)&r^{-1} + O(r)     \end{pmatrix}
% \begin{pmatrix}  \tilde{\ol B}_{00} \\ \ol g^{AB}\ol{ \nabla_0B_{AB}} \\ \tilde{\ol {\nabla_0 A_{00}} } \\ \ol g^{AB}\ol{\nabla_{0} \nabla_{0}A_{AB}}
%\end{pmatrix} =0
% \;.
%\end{eqnarray*}
%
\begin{eqnarray*}
 \left[ \partial_1
  + r^{-1} \begin{pmatrix}   3&-\frac{1}{2} &0 &0 \\ 2 & 1 &0 &0 \\ 2 & 0 &3 &0\\ 0&2&  2& 2     \end{pmatrix}
+ M
\right]
v
 =0
 \;,
\end{eqnarray*}
with $M=O(r)$ being some matrix,
and
\begin{equation}
 v:=  \begin{pmatrix}  \tilde{\ol B}_{00} \\ \ol g^{AB}\ol{ \nabla_0B_{AB}} \\ \tilde{\ol {\nabla_0 A_{00}} } \\ \tilde{\ol{ g^{AB}\nabla_{0} \nabla_{0}A_{AB}}}
\end{pmatrix}
=O(1)
\;.
\end{equation}
The change of basis matrix
\begin{equation*}
T:= \begin{pmatrix}
0 & -1/3 & -1/3 &0 \\ 0  & -2/3 &0 &0\\ 1/\sqrt{5} &2/3 & 4/3 & 0 \\ 2/\sqrt{5} & 8/3 & 0 &1
\end{pmatrix}
\;,
\end{equation*}
transforms the indicial matrix to Jordan normal form, and we end up with another Fuchsian ODE-system,
\begin{eqnarray}
 \partial_1\tilde v
  +r^{-1} \begin{pmatrix}   3&0 &0 &0 \\ 0 &2&1 &0 \\ 0& 0 & 2 &0\\ 0& 0&  0& 2     \end{pmatrix}
\tilde v
+ \tilde M \tilde v=0\;,
\label{Fuchsian_system_2}
\end{eqnarray}
where  $ \tilde v:= T^{-1}v=O(1)$ and $\tilde M:=T^{-1}MT= O(r)$.
Again, Lemma~\ref{lemma_Fuchsian} in Appendix~\ref{app_Fuchsian} (with $\lambda=-1$) implies
\begin{equation*}
 \tilde v= 0
\;,
\end{equation*}
and  thus $\ol B_{00}=\ol g^{AB}\ol{\nabla_0 B_{AB}}=\ol{\nabla_0 A_{00}}=\ol g^{AB}\ol{\nabla_0\nabla_0 A_{AB}}=0$.

In this section we have proved:
\begin{lemma}
\label{lemma_B0mu}
 Assume that \eq{wave_X0} and \eq{wave_Y0} hold, and that $\ol A_{\mu\nu} = 0 = \ol B_{ij}=\ol{\nabla_0 A_{ij}}$.
Then $\ol B_{0\mu}=0$, $ \ol{\nabla_0B_{1i}} = \ol g^{AB} \ol{ \nabla_{0} B_{AB}}=0$,  $\ol{\nabla_0 A_{0\mu}} =0$ and $\ol{\nabla_0\nabla_0 A_{1i}} = \ol g^{AB}\ol{\nabla_0\nabla_0 A_{AB}}=0$.
\end{lemma}

\subsubsection{Stronger version}

As a straightforward consequence of Theorem~\ref{KID_eqns_main} and the preceding considerations which led us to Lemma~\ref{lemma_A0mu} and \ref{lemma_B0mu}
we end up with the following result:
\begin{theorem}
 \label{KID_eqns_main_cone}
Assume that we have been given a $3+1$-dimensional spacetime $(\mcM, g, \Theta)$, with ($g, \Theta$) a smooth solution of the conformal field
equations~\eq{conf1}-\eq{conf6}.
Let $C_O\subset  \mcM$ be a light-cone.
 Then there exists a vector field $\hat X$
satisfying the unphysical Killing equations \eq{2_conditions} on $\mathrm{D}^+(C_O)$
if and only if there exists a pair $(X,Y)$, $X$ a vector field and $Y$ a function, which fulfills the following conditions:
 \begin{enumerate}
  \item[(i)] $\Box_g X_{\mu} + R_{\mu}{}^{\nu}X_{\nu}  + 2\nabla_{\mu}Y= 0$,
  \item[(ii)] $ \Box_g Y+ \frac{1}{6}X^{\mu}\nabla_{\mu}R+ \frac{1}{3} R Y=0$,
   \item[(iii)] $ \ol\phi = 0$ with $ \phi \equiv   X^{\mu}\nabla_{\mu}\Theta - \Theta Y  $,
  \item[(iv)] $\ol\psi=0$ with  $\psi \equiv  X^{\mu}\nabla_{\mu}s    +sY -\nabla_{\mu}\Theta \nabla^{\mu}Y $,
  \item[(v)] $\ol A_{ij} =0$  with
 $A_{\mu\nu} \equiv \nabla_{\mu}X_{\nu} + \nabla_{\nu} X_{\mu} - 2Y g_{\mu\nu}$,
  \item[(vi)] $\ol A_{01}=0$,
  \item[(vii)] $\ol B_{ij}=0$  with
$B_{\mu\nu}\equiv \mcL_XL_{\mu\nu} +   \nabla_{\mu}\nabla_{\nu}Y$.
 \end{enumerate}
Moreover, $\hat X=X$ and $\nabla_{\kappa}\hat X^{\kappa}= 4 Y$.
 The condition (vi) is not needed on the closure of those sets where $\tau$ is non-zero.
\end{theorem}

%\begin{remark}
 %The transverse derivatives of $X$ and $Y$ appearing in (iv), (vi) and (vii) can be eliminated via (i) and (ii).
%\end{remark}

\subsubsection{The (proper) KID equations}

The conditions (iv), (vi) and (vii) in Theorem~\ref{KID_eqns_main_cone} are not intrinsic in the sense that they involve
transverse derivatives of $X$ and $Y$ which are not part of the initial data for the wave equations (i) and (ii). However, they can be eliminated via these wave equations.
In fact, this is useful if one wants to check for a certain candidate field defined only on the initial surface whether it extends to a vector field satisfying the
unphysical Killing equations or not.
In essence this is what we will do next.

We have
\begin{eqnarray}
\ol{ \Box_g Y}  &\equiv& 2\nu^0( \nabla_1+\frac{1}{2}\tau )\ol{\nabla_0 Y}  + \ol g^{ij} \coneD_i\coneD_j \ol Y
\;,
\label{wave_coneD}
\end{eqnarray}
where $D_{i}$ is the derivative operator introduced in \cite{CP1},
\begin{eqnarray*}
 \coneD_i \ol Y &:=& \nabla_i \ol Y  \;,
\\
\coneD_i \ol X_{\mu} &:=& \nabla_i \ol X_{\mu}\;,
\\
 \coneD_{i} \coneD_j \oY
  &:= & \partial_i  \coneD_j \oY - \ol\Gamma^k_{ij}\coneD_k Y
 \;,
\\
 \coneD_{i} \coneD_j \ol X_ \mu
 &: = &  \partial_ i \coneD_j \ol X_\mu - \ol\Gamma^k_{ij} \coneD_k \ol Y_\mu - \ol\Gamma^\nu_{i\mu} \coneD_j \ol Y_\nu
 \;,
\end{eqnarray*}
i.e.\ one simply removes the transverse derivatives which would appear in the corresponding expressions with covariant derivatives.
Since the action of $\nabla_i$ and $\coneD_i$ coincides in many cases relevant to us one may often use them interchangeably.
Nevertheless, we shall use $\coneD_i$ consistently whenever derivatives of $X$ or $Y$ appear
in order to stress that no transverse derivatives of these fields are involved.

By (ii) and \eq{wave_coneD} the function $\Upsilon:= \ol{\partial_0 Y}$ (note that $\Upsilon$ is not a scalar) satisfies the ODE
\begin{eqnarray}
 ( \partial_1 +\frac{\tau }{2}-\ol \Gamma^0_{01})\Upsilon
-\ol \Gamma^{i}_{01}\nabla_{i}\ol Y
  + \frac{1}{2}\nu_0\big(\ol g^{ij} \coneD_i\coneD_j \ol Y + \frac{1}{6}\ol X^{\mu}\ol{\nabla_{\mu}R}+ \frac{1}{3} \ol R \ol Y\big) =0
\;.
\label{ODE_rho}
\end{eqnarray}
Regularity requires $\Upsilon =O(1)$.

It is useful to make the following definition
\begin{equation}
 S_{\mu\nu\sigma}:= \nabla_{\mu}\nabla_{\nu} X_{\sigma}  -  R_{\sigma\nu\mu}{}^{\kappa} X_{\kappa}  - 2\nabla_{(\mu}Y g_{\nu)\sigma} +  \nabla_{\sigma}Y g_{\mu\nu}
 \;.
\label{dfn_tensor_S}
\end{equation}
It follows from the identity \eq{relation_nabla_A} that
\begin{equation}
  S_{\mu\nu\sigma}=  \nabla_{(\mu}A_{\nu)\sigma} - \frac{1}{2}\nabla_{\sigma}A_{\mu\nu}
\;.
\label{rln_tensor_S}
\end{equation}
Note that this implies the useful relations
\begin{eqnarray}
 2S_{\mu(\nu\sigma)} &=& \nabla_{\mu}A_{\nu\sigma}
\;,
\label{useful_rln_S_1}
\\
S_{[\mu\nu]\sigma} &=& 0
\;.
\label{useful_rln_S_2}
\end{eqnarray}
Recall that \eq{relation_for_A01} is a consequence of (i) and (v). Hence
\begin{eqnarray}
 \ol S_{110} &=&  \nabla_{1}\ol A_{01} - \frac{1}{2}\ol{\nabla_{0}A_{11}}
\nonumber
\\
 &=& (\partial_{1}+\frac{1}{2}\tau  - \nu^0\partial_1\nu_0) \ol A_{01}
 \;.
\label{S110_A01}
\end{eqnarray}
We conclude that due to regularity  we have, assuming (i) and (v),
\begin{equation*}
 \ol A_{01}=0 \quad \Longleftrightarrow \quad \ol S_{110}=0
\;.
\end{equation*}
This leads us to the following stronger version of Theorem~\ref{KID_eqns_main_cone}:
\begin{theorem}
\label{KID_eqns_main_cone2}
Assume that we have been given a $3+1$-dimensional spacetime $(\mcM, g, \Theta)$, with ($g, \Theta$) being a smooth solution of the conformal field
equations~\eq{conf1}-\eq{conf6}.
Let $\mathring X$ be a vector field and $\mathring Y$ a function defined on a light-cone $C_O\subset \mcM$.
Then there exists a smooth vector field $X$ with $\ol X=\mathring X$ and $\ol{\nabla_{\kappa}X^{\kappa}} = 4\mathring Y$ satisfying the unphysical Killing equations  \eq{2_conditions} on $\mathrm{D}^+(C_O)$
(i.e.\ representing a Killing field of the physical spacetime)
if and only if
 \begin{enumerate}
   \item[(a)] the conditions (iii) and (v) in Theorem~\ref{KID_eqns_main_cone} hold,
  \item[(b)] $\ol{\psi}^{\mathrm{intr}} :=  \mathring X^{\mu}\ol{\nabla_{\mu}s  }  +\ol s\mathring Y -  \ol{\nabla^{i}\Theta} \coneD_{i}\mathring Y  - \nu^0 \Upsilon\nabla_{1}\ol \Theta=0$,
  \item[(c)] $\ol S_{110} \equiv \coneD_{1}\coneD_{1}\mathring X_{0}  - \ol  R_{011}{}^{\kappa} \mathring X_{\kappa}  - 2\nu_0\coneD_{1}\mathring Y =0$,
  \item[(d)] $\ol B_{1i} \equiv  \mathring X^{\kappa}\nabla_{\kappa}\ol L_{1i} +2 \ol L_{\kappa(1}\coneD_{i)}\mathring X^{\kappa}+   \coneD_{1}\coneD_{i}\mathring Y=0$,
  \item[(e)] $\ol { B}^{\mathrm{intr}}_{AB} := \mathring X^{\kappa}\nabla_{\kappa}\ol L_{AB} +2 \ol L_{\kappa(A}\coneD_{B)}\mathring X^{\kappa}+   \coneD_{A}\coneD_{B}\mathring Y +\nu^0\Upsilon\chi_{AB}=0$,
\item[(f)] $\mathring X$ and $\mathring Y$ are restrictions to the light-cone of smooth spacetime fields.
 \end{enumerate}
% \begin{enumerate}
%   \item[(i)] $ \ol\phi\equiv   \mathring X^{\mu}\ol{\nabla_{\mu}\Theta} - \ol \Theta \mathring Y=0$,
%  \item[(ii)] $\ol{\psi}^{intr} :=  \mathring X^{\mu}\ol{\nabla_{\mu}s  }  +\ol s\mathring Y -  \ol{\nabla^{i}\Theta} \coneD_{i}\mathring Y  - \nu^0 \Upsilon\nabla_{1}\ol \Theta=0$,
%  \item[(iii)] $\ol A_{ij} \equiv  \coneD_{i}\mathring X_{j} + \coneD_{j} \mathring X_{i} - 2\mathring Y \ol g_{ij}=0$,
%  \item[(iv)] $\ol S_{110} \equiv \coneD_{1}\coneD_{1}\mathring X_{0}  - \ol  R_{011}{}^{\kappa} \mathring X_{\kappa}  - 2\nu_0\coneD_{1}\mathring Y =0$,
%  \item[(v)] $\ol B_{1i} \equiv  \mathring X^{\kappa}\nabla_{\kappa}\ol L_{1i} +2 \ol L_{\kappa(1}\coneD_{i)}\mathring X^{\kappa}+   \coneD_{1}\coneD_{i}\mathring Y=0$,
 % \item[(vi)] $\ol { B}^{intr}_{AB} := \mathring X^{\kappa}\nabla_{\kappa}\ol L_{AB} +2 \ol L_{\kappa(A}\coneD_{B)}\mathring X^{\kappa}+   \coneD_{A}\coneD_{B}\mathring Y +\nu^0\Upsilon\chi_{AB}=0$,
%\item[(vii)] $\mathring X$ and $\mathring Y$ are restrictions to the light-cone of smooth spacetime fields,
% \end{enumerate}
The function $\Upsilon$ is the unique solution of
\begin{eqnarray}
 ( \partial_1 +\frac{\tau }{2}-\ol \Gamma^0_{01})\Upsilon
-\ol \Gamma^{i}_{01}\coneD_{i}\mathring Y
  + \frac{1}{2}\nu_0\big(\ol g^{ij} \coneD_i\coneD_j \mathring Y + \frac{1}{6}\mathring X^{\mu}\ol{\nabla_{\mu}R}+ \frac{1}{3} \ol R \mathring Y\big) =0
\label{thm_rho_eqn}
\end{eqnarray}
which is bounded near the tip of the cone.
 The condition (c) is not needed on the closure of those sets on which the expansion $\tau$ is non-zero.
\end{theorem}

\begin{proof}
It needs to be shown that $(\mathring X,\mathring Y)$ extends to a pair $(X,Y)$ satisfying  (i)-(vii) in Theorem~\ref{KID_eqns_main_cone}. From the considerations
above it becomes clear that (a)-(e)  do imply (i)-(vii) in Theorem~\ref{KID_eqns_main_cone} if $\mathring X$ and $\mathring Y$ can be extended
to smooth solutions of the wave equations \eq{wave_X0} and \eq{wave_Y0} for $X$ and $Y$. However, this follows from \cite{dossa} due to (f).
\qed
\end{proof}

\begin{remark}
{\rm
The conditions (a)-(e) will be called  \textit{(proper)%
\footnote{in the sense of ``intrinsic'', they do not involve transverse derivatives of $X$ or $Y$}
 Killing Initial Data (KID) equations}
(cf.\ Proposition~\ref{KID_eqns_main_cone3} below which shows that condition (f) is not needed).
}
\end{remark}

\begin{remark}
{\rm
Theorem~\ref{KID_eqns_main_cone2} can e.g.\ be applied to a light-cone with vertex at past timelike infinity for vanishing cosmological constant
(this is done in Section~\ref{sec_special_cone}),
or to light-cones with vertex on $\scri^-$ for vanishing or positive cosmological constant.
}
\end{remark}

\subsubsection{Extendability of the candidate fields}
\label{extendability}

A drawback of Theorem~\ref{KID_eqns_main_cone2} is the condition (f): Usually it is a non-trivial issue to make sure that the candidate fields $\mathring X$ and $\mathring Y$ which are constructed from (a subset of) (a)-(e)  are  restrictions to the light-cone of smooth spacetime fields.
(Nonetheless we shall see in Section~\ref{sec_special_cone} that (f) becomes trivial on the $C_{i^-}$-cone.)
We therefore aim to prove that (f) follows
 directly and without any restrictions from the KID equations (a)-(e).
%from the remaining conditions.
%We conclude this section with a remark.
%In \cite{CP1} it was possible to deduce that the candidate field is the restriction to the cone of a smooth spacetime field directly and without any restrictions
%from the KID equations.
%In this case this seems to be harder to prove in full generality
%maybe one should consider the $\ol B_{11}=0$-equation?}
%due to the fact that (i)-(vi) do not provide a ``nice'' equation for $\mathring Y$. In fact it depends on the type of the initial surface which equation determines %$\mathring Y$: If e.g.\ $ \Theta=1$ it is (i), if $\ol \Theta=0$ and $\ol s=\mathrm{const}\ne 0 $ it is (ii),
%  and both case are of relevance, cf.\ Section~\ref{sec_special_case} and
%Section~\ref{sec_special_cone} below.

%Our starting point is Theorem~\ref{KID_eqns_main_cone} and the hypotheses (i)-(vii) made therein.
%Recall the identity \eq{rln_tensor_S} and that \eq{rln_nabla0_A1A} is a consequence of (i, (v) and (vi).
%We thus have
%%
%\begin{equation}
% \ol S_{A10} = \nabla_{(A}\ol A_{1)0} - \frac{1}{2}\ol{\nabla_0 A_{1A}} = (\partial_1  -\ol \Gamma^0_{01})\ol A_{0A}  -\sigma_A{}^B\ol A_{0B}
%\;.
%\end{equation}
%%
% We have shown above that $\ol A_{0A}=0$ follows from (i)-(vii), hence
%%
%\begin{equation}
% \ol S_{A10} =0
%\;.
%\end{equation}

Since the validity of (f) is only non-trivial in some neighbourhood of the vertex of the cone, we can and will assume in this section that the
expansion $\tau$ has no zeros.

The proceeding will be in close analogy to  \cite[Section~2.5]{CP1}.
First we want to compute the divergence $\nabla^{\alpha}S_{\alpha\beta\gamma}$ which contains certain transverse derivatives
of $\mathring X$ and $\mathring Y$ (which eventually drop out from the relevant formulae).
For these expressions to make sense let $X$ and $Y$ be any smooth extensions of  $\mathring X$ and $\mathring Y$  from the cone $C_O$ to a punctured
neighbourhood of $O$. We stress that no assumptions are made concerning the behaviour of $X$ and $Y$ as the tip of the cone is approached.

By \eq{dfn_tensor_S} and the second Bianchi identity we have
\begin{eqnarray}
 \nabla^{\sigma}S_{\mu\nu\sigma} \equiv \frac{1}{2} \nabla_{\mu} \nabla_{\nu}A_{\sigma}{}^{\sigma} + 2 B_{\mu\nu}
+ \frac{1}{6} RA_{\mu\nu}
+(B_{\sigma}{}^{\sigma}-L^{\alpha\beta}A_{\alpha\beta}) g_{\mu\nu}
\;.
\label{rln_div_S}
\end{eqnarray}
In adapted null coordinates the trace of the left-hand side of \eq{rln_div_S} on the cone reads,
\begin{eqnarray}
 \ol{\nabla^{\sigma}S_{\mu\nu\sigma}} &=& \nu^0(\ol{\nabla_0 S_{\mu\nu 1}} + \nabla_1 \ol S_{\mu\nu 0} )
+ \ol g^{1A}(\nabla_1 \ol S_{\mu\nu A} +  \nabla_A \ol S_{\mu\nu 1})
\nonumber
\\
&& + \ol g^{11}\nabla_1\ol S_{\mu\nu 1} + \ol g^{AB} \nabla_A\ol S_{\mu\nu B}
\;.
\label{rln_div_S2}
\end{eqnarray}
The undesirable transverse derivatives which appear in $\ol{\nabla_0 S_{\mu\nu 1}}$ can be eliminated via
\begin{eqnarray}
 \nabla_0 \nabla_{\mu}\nabla_{\nu}X_{\sigma} &=&
\nabla_{\mu}\nabla_{\nu}A_{0\sigma}
- \nabla_{\mu}\nabla_{\nu}\nabla_{\sigma} X_{0}
+2g_{0\sigma}\nabla_{\mu}\nabla_{\nu} Y
+ \nabla_{\mu}(R_{0\nu\sigma}{}^{\kappa}X_{\kappa})
\nonumber
\\
&&
 + R_{0\mu \nu}{}^{\kappa}\nabla_{\kappa}X_{\sigma} + R_{0\mu \sigma}{}^{\kappa}\nabla_{\nu}X_{\kappa}
\;.
\label{nabla3_1}
\end{eqnarray}

\begin{lemma}
\label{lemma_S110}
 Assume $\ol A_{ij}=0$. Then
$$
 2 \ol B_{11}
=  \tau\nu^0 \ol S_{11 0}
\;.
$$
\end{lemma}

\begin{proof}
Equation \eq{rln_div_S} with $(\mu\nu)=(11)$ yields
\begin{eqnarray}
 \ol{\nabla^{\sigma}S_{11\sigma}} = \nu^0 \nabla_{1} \nabla_{1}\ol A_{01} + 2 \ol B_{11}
\;.
\label{div_S_11_1}
\end{eqnarray}
Note that it follows from \eq{rln_tensor_S} that the vanishing of $\ol A_{1i}$ implies the vanishing
of  $\ol S_{11i}$ as well as all permutations thereof.
Due to \eq{rln_div_S2} we further have
\begin{eqnarray*}
\ol S_{1A B} = \ol S_{A1B}=0
\;.
\end{eqnarray*}
From \eq{rln_div_S2} we then obtain with $(\mu\nu)=(11)$
\begin{eqnarray}
 \ol{\nabla^{\sigma}S_{11\sigma}} &=& \nu^0\ol{\nabla_0 S_{11 1}} + \nu^0\nabla_1 \ol S_{11 0}
 - 2\chi^{AB}\ol S_{1A B}+ \tau\nu^0 \ol S_{11 0}
\;,
\label{div_S_11_2}
\end{eqnarray}
while \eq{nabla3_1} gives
\begin{eqnarray*}
 \ol{\nabla_0 S_{111}} = \ol{\nabla_0 \nabla_{1}\nabla_{1} X_{1} }
= \nabla_{1}\nabla_{1}\ol A_{01}
- \nabla_{1}\ol S_{110}
\;.
\end{eqnarray*}
Equating \eq{div_S_11_1} with \eq{div_S_11_2} yields the desired result.
%%
%\begin{equation}
% 2 \ol B_{11}
%=  \tau\nu^0 \ol S_{11 0}
%\;,
%\label{rln_S110_B11}
%\end{equation}
%%
%which finishes the proof.
\qed
\end{proof}

\begin{lemma}
\label{lemma_SA10}
 Assume $\ol A_{ij}=0$ and $\ol S_{110}=0$. Then
$$
 2 \nu_0\ol B_{1A}
= (\partial_1  +\tau
- \nu^0\partial_1\nu_0) \ol S_{A1 0}
\;.
$$
\end{lemma}

\begin{proof}
From the $(\mu\nu)=(A1)$-components of \eq{rln_div_S} we deduce
\begin{eqnarray}
 \ol{\nabla^{\sigma}S_{A1\sigma}} &=&
 \frac{1}{2} \nabla_{A}\nabla_{1} \ol A_{\sigma}{}^{\sigma}
+ 2 \ol B_{1A}
\;.
\label{div_S_A1_1}
\end{eqnarray}
It follows from \eq{nabla3_1} that
\begin{eqnarray*}
 \ol{\nabla_0 S_{A11}} &=&
\ol{\nabla_0 \nabla_{A}\nabla_{1} X_{1}}
\,=\, \nabla_{A}\nabla_{1}\ol A_{01}
- \nabla_{A}\ol S_{110}
\\
&=& \nabla_{A}\nabla_{1}\ol A_{01}
+ 2\chi_A{}^B\ol S_{B10}
\;.
\end{eqnarray*}
Recall that  $\ol S_{11i}$ as well as all permutations thereof vanish, and that $\ol S_{A1B}=\ol S_{1AB}=0$.
Equation \eq{rln_div_S2} then yields with $(\mu\nu)=(A1)$
\begin{eqnarray}
 \ol{\nabla^{\sigma}S_{A1\sigma}} &=&
\nu^0\ol{\nabla_0 S_{A1 1}} + \nu^0\nabla_1 \ol S_{A1 0} + \ol g^{1B}\nabla_B\ol S_{A11}
 + \ol g^{BC} \nabla_C\ol S_{A1 B}
\nonumber
\\
&=&  \nu^0\nabla_{A}\nabla_{1}\ol A_{01}
+ 2\nu^0\chi_A{}^B\ol S_{B10} + \nu^0\nabla_1 \ol S_{A1 0}
 + \frac{1}{2}\ol g^{1B}\ol{\nabla_B\nabla_A A_{11}}
\nonumber
\\
&&
 + \ol g^{BC} \nabla_C\ol S_{A1 B}
\;,
\label{div_S_A1_2}
\end{eqnarray}
where we also used \eq{useful_rln_S_1}.
Combining \eq{div_S_A1_1} and \eq{div_S_A1_2} and invoking again \eq{useful_rln_S_1} we obtain
\begin{eqnarray}
 2 \ol B_{1A}
&=& \nu^0\nabla_1 \ol S_{A1 0}+ 2\nu^0\chi_A{}^B\ol S_{B10}
 + \ol g^{1B}(\frac{1}{2}\ol{\nabla_B\nabla_A A_{11}}- \nabla_{A}\nabla_{1} \ol A_{1B})
\nonumber
\\
&&
 + 2\ol g^{BC}\nabla_{[C}\ol S_{A]1 B}
\;.
\end{eqnarray}
Since
\begin{eqnarray*}
 \frac{1}{2}\ol{\nabla_B\nabla_A A_{11}}- \nabla_{A}\nabla_{1} \ol A_{1B} \,=\, -\nabla_A\ol S_{11B} \,=\, 0
\;,
\end{eqnarray*}
and
\begin{eqnarray*}
 2\ol g^{BC}\nabla_{[C}\ol S_{A]1 B} \,=\,  \tau \nu^0\ol S_{A1 0} - \nu^0\chi_{A}{}^B\ol S_{B1 0}
 +\underbrace{2\ol g^{BC}\chi_{[A}{}^D\ol S_{C]D B}}_{=0 \text{ by \eq{rln_tensor_S}}}
\end{eqnarray*}
the lemma is proved.
\qed
\end{proof}

As in \cite{CP1} one checks via the formulae in \cite[Section~4.5]{CCM2}
which hold in any sufficiently regular gauge,
and assuming
\begin{eqnarray*}
 &\mathring X_1,\partial_i\mathring X =O(1)\;, \quad \mathring X_0,\partial_i\mathring X_0,\partial_A\partial_1\mathring X_0 =O(1)\;,&
\\
&\mathring X_A, \partial_B\mathring X_A = O(r), \quad \partial_1\mathring X_A=O(1),\quad
\mathring Y, \partial_i \mathring Y = O(1)
\;,&
\end{eqnarray*}
which is necessarily satisfied by any pair $(\mathring X,\mathring Y)= (\ol X, \frac{1}{4}\,\ol{\mathrm{div} X})$ with $X$ a smooth vector field,
% satisfying the unphysical Killing equations,
that $\ol S_{A10}$ needs to exhibit the following behaviour near the tip of the cone:
\begin{equation}
 \ol S_{A10} = O(r^{-1})
\;.
\end{equation}
It thus follows immediately from Lemma~\ref{lemma_S110} and \ref{lemma_SA10} that for any vector field $\mathring X$
and any function $\mathring Y$
%which are bounded at the vertex of the cone and
which satisfy $\ol  A_{ij}=0$ and $\ol B_{1i}=0$ the equations
\begin{equation}
 \ol S_{i10}=0
\end{equation}
hold sufficiently close to the vertex of the cone where $\tau$ has no zeros.

Let us define the antisymmetric tensor field $\mathring F_{\mu\nu}$ via
\begin{eqnarray}
 \mathring F_{ij} &:=& \nabla_{[i}\mathring X_{j]}
\;,
\\
\mathring F_{i0} &:=& \nabla_i\mathring X_0 - \ol g_{0i} \mathring Y
\;.
\end{eqnarray}
We also define the covector field $\mathring H_{\mu}$,
\begin{eqnarray}
 \mathring H_i &:=& \nabla_i \mathring Y
\;,
\\
\mathring H_0 &:=& 0
\;.
\end{eqnarray}
%
%where $\Upsilon$ was defined in \eq{thm_rho_eqn}.

For the following computations we assume
\begin{equation}
 \ol S_{i10} = 0 = \ol A_{ij} = \ol B_{1i}
\;.
\end{equation}
Then, due to the first Bianchi identity,
\begin{eqnarray*}
 \mathring F_{1i} &\equiv& \nabla_1 \mathring X_i - \frac{1}{2}\ol A_{1i} -\mathring Y \ol g_{1i} \,=\,  \nabla_1 \mathring X_i
\;,
\\
\nabla_1 \mathring F_{ij} &\equiv & \nabla_{[i}\ol A_{j]1} - 2\ol g_{1[i}\mathring H_{j]}
  - \ol R_{ij1}{}^{\alpha}\mathring X_{\alpha}  \,=\, - \ol R_{ij1}{}^{\alpha}\mathring X_{\alpha}
\;,
\\
 \nabla_1\ol F_{i0} &\equiv & \ol S_{i10} - \ol R_{i01}{}^{\alpha}\mathring X_{\alpha}
+ \nu_0\mathring H_{i}  -\ol g_{1i} \ol {\nabla_0 Y}
\,=\,\nu_0\mathring H_{i}   - \ol R_{i01}{}^{\alpha}\mathring X_{\alpha}
\;.
\end{eqnarray*}
Moreover,
\begin{eqnarray*}
 \nabla_1 \mathring H_i &\equiv &\ol B_{1i}
 - \ol{\mcL_X L_{1i}}
 \,\equiv \,
\ol B_{1i} - \ol L_{(1}{}^{j}\ol A_{i)j}
 -\mathring X^{\alpha}\ol{\nabla_{\alpha} L_{1i}}
-2\ol L_{(1}{}^{\alpha}\mathring F_{i)\alpha}
-2\ol L_{1i}\mathring Y
\\
&=&  -\mathring X^{\alpha}\ol{\nabla_{\alpha} L_{1i}}
-2\ol L_{(1}{}^{\alpha}(\mathring F_{i)\alpha}
+\ol g_{i)\alpha}\mathring Y )
\;,
\\
\nabla_1 \mathring H_0 &\equiv &  -\ol \Gamma^{i}_{01}\mathring H_{i}
\;.
\end{eqnarray*}

Therefore the candidate fields $\mathring X$ and $\mathring Y$ solving (a)-(e) in Theorem~\ref{KID_eqns_main_cone2} form a solution of the following problem
on $C_O$,
%(using obvious notation $\mathring F_{\mu\nu}$ and $\mathring H_{\mu}$),
%
\begin{eqnarray}
 \left\{
   \begin{array}{l}
     \nabla_1\mathring  X_{\mu} = \mathring  F_{1\mu} + \ol  g_{1\mu}\mathring  Y  ,  %& \hbox{on $C_O$;}
\\
    \nabla_1\mathring  F_{\mu\nu} =  2\ol  g_{1[\nu}\mathring  H_{\mu]} - \ol R_{\mu\nu 1}{}^{\alpha}\mathring  X_{\alpha}  , %& \hbox{on $C_O$;}
 \\
    \nabla_1 \mathring  Y = \mathring  H_1  ,% & \hbox{on $C_O$;}
\\
 \nabla_1\mathring  H_{\mu} =
 -\mathring  X^{\alpha}\ol{\nabla_{\alpha} L_{1\mu}}
-2\ol L_{(1}{}^{\alpha}(\mathring  F_{\mu)\alpha}
+\ol g_{\mu)\alpha}\mathring  Y ) %&
\\
\phantom{\nabla_1\mathring  H_{\mu} =}-\ol g_{1\mu}\nu^0[\ol\Gamma^{i}_{01}\mathring  H_{i}  -\mathring  X^{\alpha}\ol{\nabla_{\alpha} L_{01}}
-2\ol L_{(1}{}^{\alpha}(\mathring  F_{0)\alpha}
+\ol g_{0)\alpha}\mathring  Y )] , %& \hbox{on $C_O$.}
   \end{array}
 \right.
\label{system1}
\end{eqnarray}
which is uniquely defined by the values of $\mathring X_{\mu}$, $\mathring F_{\mu\nu}$, $\mathring Y$ and $\mathring H_{\mu}$ at the vertex of the cone.

We want to show that the fields which solve \eq{system1} are restrictions to the cone of smooth spacetime fields:
Given  any
vector  $\ell^{\mu}$  in the tangent space
at $O$ define $(x^{\mu}(s), X^{\mu}(s), F_{\mu\nu}(s),Y(s),H_{\mu}(s))$ as the unique solution of the problem
\begin{eqnarray}
\left\{
   \begin{array}{l}
\frac{\mathrm{d}^2x^{\mu}}{\mathrm{d}s^2} + \Gamma^{\mu}_{\alpha\beta}\frac{\mathrm{d}x^{\alpha}}{\mathrm{d}s}\frac{\mathrm{d}x^{\beta}}{\mathrm{d}s} =0\;,
\\
\frac{\mathrm{d}X_{\mu}}{\mathrm{d}s} - \Gamma^{\alpha}_{\mu\beta}X_{\alpha}\frac{\mathrm{d}x^{\beta}}{\mathrm{d}s} =
 F_{\alpha\mu}\frac{\mathrm{d}x^{\alpha}}{\mathrm{d}s}+ g_{\alpha\mu}Y\frac{\mathrm{d}x^{\alpha}}{\mathrm{d}s}
\;,
\\
\frac{\mathrm{d}F_{\mu\nu}}{\mathrm{d}s} - \Gamma^{\alpha}_{\mu\gamma}F_{\alpha\nu}\frac{\mathrm{d}x^{\gamma}}{\mathrm{d}s}
- \Gamma^{\alpha}_{\nu\gamma}F_{\mu\alpha}\frac{\mathrm{d}x^{\gamma}}{\mathrm{d}s} =
2 g_{\gamma[\nu}  H_{\mu]} \frac{\mathrm{d}x^{\gamma}}{\mathrm{d}s}
-R_{\mu\nu\gamma}{}^{\alpha}X_{\alpha}\frac{\mathrm{d}x^{\gamma}}{\mathrm{d}s}
\;,
\\
\frac{\mathrm{d}Y}{\mathrm{d}s} = H_{\alpha}\frac{\mathrm{d}x^{\alpha}}{\mathrm{d}s}
\;,
\\
\frac{\mathrm{d}H_{\mu}}{\mathrm{d}s} - \Gamma^{\alpha}_{\mu\beta}H_{\alpha}\frac{\mathrm{d}x^{\beta}}{\mathrm{d}s} =
\{ -  X^{\alpha}\nabla_{\alpha} L_{\gamma\mu}
-2 L_{(\gamma}{}^{\alpha}(  F_{\mu)\alpha}
+ g_{\mu)\alpha}  Y ) %&
\\
\qquad - g_{\gamma\mu}\nu^0[\Gamma^{i}_{01}  H_{i}  -  X^{\alpha}\nabla_{\alpha} L_{01}
-2 L_{(1}{}^{\alpha}(  F_{0)\alpha}
+ g_{0)\alpha}  Y )] \}\frac{\mathrm{d}x^{\gamma}}{\mathrm{d}s}
\;,
\\
x^{\mu}(0)=0\;, \quad \frac{\mathrm{d}x^{\mu}}{\mathrm{d}s} (0) =\ell^{\mu}\;,
   \end{array}
 \right.
\label{system2}
\end{eqnarray}
for given initial data $(X^{\mu}(0), F_{\mu\nu}(0),Y(0),H_{\mu}(0))$.
As in \cite[Section~2.4]{CP1}  the system \eq{system2}, together with the property that solutions of ODEs depend smoothly upon initial data, and that the trace of solutions of \eq{system2} on $C_O$ solve \eq{system1}, can be used to show that the fields solving \eq{system1}  are restrictions to the cone of smooth spacetime fields.
We have proved:

\begin{proposition}
\label{KID_eqns_main_cone3}
The condition (f) in Theorem~\ref{KID_eqns_main_cone2} can be removed.
\end{proposition}

\subsection{A stronger version of Theorem~\ref{KID_eqns_main} for two transversally intersecting null hypersurfaces}

\subsubsection{Stronger version}

We want to establish the analogues of Lemma~\ref{lemma_A0mu} and \ref{lemma_B0mu} for two transversally intersecting null hypersurfaces.
\begin{lemma}
\label{lemma_A0mu_hyper}
 Assume that the wave equations \eq{wave_X0} and \eq{wave_Y0} for $X$ and $Y$ hold, and that, on $N_1$,  $\ol A_{2\mu} =  \ol A_{AB} = 0 = \ol B_{22} = \ol B_{2A} = \ol B_{AB}$, similarly on $N_2$.
Furthermore, we assume that $\nabla_{[1}A_{2]A}|_S=0$.
Then, on $N_1$, $\ol A_{11}=\ol A_{1A} = 0 $ and $\ol{\nabla_1 A_{22}} = \ol{\nabla_1 A_{2A}}= \ol{\nabla_1 A_{AB}}=0 $, and a corresponding statement holds on $N_2$.
On the closure of those sets where $\tau$ is non-zero the assumption $\ol A_{12}=0$ is not needed
 but follows from the remaining assumptions,
supposing that $ A_{12}|_S=0$.
\end{lemma}

\begin{proof}
We can repeat most of the steps which were necessary to prove Lemma~\ref{lemma_A0mu}. The only difference is that the ODEs are not of Fuchsian type anymore, but regular ones. To make sure that all the fields involved vanish on $N_1\cup N_2$ we therefore need to make sure that we have vanishing initial data on $S$. This is the case if, \underline{on $S$},
\begin{eqnarray*}
 A_{11} = A_{22} = A_{1A} = A_{2A} = \nabla_1 A_{2A} = \nabla_2 A_{1A} = g^{AB}\nabla_1 A_{AB} = g^{AB} \nabla_2 A_{AB}=0
 \;.
\end{eqnarray*}
Observing that the analogue of \eq{rln_nabla0_A1A} for light-cones holds, i.e.\
\begin{eqnarray*}
 \ol{\nabla_{(1}A_{2)A}}=0
\;.
\end{eqnarray*}
this is an obvious consequence of the hypotheses made above.
 \qed
\end{proof}

In analogy to Lemma~\ref{lemma_B0mu} we have
\begin{lemma}
\label{lemma_B0mu_hyper}
 Assume that \eq{wave_X0} and \eq{wave_Y0} hold, and that $\ol A_{\mu\nu} = 0$.
Moreover, assume that, on $N_1$,  $\ol B_{22}= \ol B_{2A}= \ol B_{AB}=0$ and $\ol{\nabla_1 A_{22}}= \ol{\nabla_1 A_{2A}} = \ol{\nabla_1 A_{AB}} =0$,
similarly on $N_2$.
Then, on $N_1$,
$\ol B_{1\mu}=0$, $ \ol{\nabla_1B_{22}}= \ol{\nabla_1B_{2A}}= \ol g^{AB} \ol{ \nabla_{1} B_{AB}}=0$,  $\ol{\nabla_1 A_{1\mu}} =0$ and $\ol{\nabla_1\nabla_1 A_{22}} = \ol{\nabla_1\nabla_1 A_{2A}} = \ol g^{AB}\ol{\nabla_1\nabla_1 A_{AB}}=0$, and  similar conclusions can be drawn on $N_2$.
\end{lemma}
\begin{proof}
 Again, we just need to make sure that all the initial data for the ODEs vanish on $S$. For all the field components involving covariant derivatives of $A_{\mu\nu}$ this follows directly from the vanishing of $\ol A_{\mu\nu}$.
The vanishing of those field components involving (covariant derivatives of) $B_{\mu\nu}$ follows from the same fact, since, by \eq{waveeqn_A}, they can be expressed in terms of $A_{\mu\nu}$ and covariant derivatives thereof.
\qed
\end{proof}

Altogether we have proved
\begin{theorem}
 \label{KID_eqns_main_hyper}
Assume that we have been given a $3+1$ dimensional spacetime $(\mcM, g, \Theta)$, with ($g, \Theta$) a smooth solution of the conformal field equations.
Let $N_a\subset \mcM$, $a=1,2$, be two transversally intersecting
null hypersurfaces  with transverse intersection along a smooth 2-dimensional submanifold $S$.
Then there exists a vector field $\hat X$ satisfying the unphysical Killing equations \eq{2_conditions} on $\mathrm{D}^+(N_1\cup N_2)$ if and only if there exists a pair $(X,Y)$, $X$ a vector field and $Y$ a function, which fulfills the following conditions:
 \begin{enumerate}
  \item[(a)] the conditions (i)-(iv) in Theorem~\ref{KID_eqns_main} hold,
  \item[(b)] $\ol A_{AB} =0 = \ol A_{22}|_{N_1}= \ol A_{2A}|_{N_1}= \ol A_{11}|_{N_2}= \ol A_{1A}|_{N_2}$
  with
 $A_{\mu\nu} \equiv \nabla_{\mu}X_{\nu} + \nabla_{\nu} X_{\mu} - 2Y g_{\mu\nu}$,
  \item[(c)] $\ol A_{12}=0$,
\item[(d)] $\ol {\nabla_{[1}A_{2]A}}|_S=0$,
  \item[(e)] $\ol B_{AB} =0 = \ol B_{22}|_{N_1}= \ol B_{2A}|_{N_1}= \ol B_{11}|_{N_2}= \ol B_{1A}|_{N_2}$   with
$B_{\mu\nu}\equiv \mcL_XL_{\mu\nu} +   \nabla_{\mu}\nabla_{\nu}Y$.
 \end{enumerate}
% \begin{enumerate}
%  \item[(i)] $\Box_g X_{\mu} + R_{\mu}{}^{\nu}X_{\nu}  + 2\nabla_{\mu}Y= 0$,
%  \item[(ii)] $ \Box_g Y+ \frac{1}{6}X^{\mu}\nabla_{\mu}R+ \frac{1}{3} R Y=0$,
%   \item[(iii)] $ \ol\phi = 0$  with $ \phi \equiv   X^{\mu}\nabla_{\mu}\Theta - \Theta Y  $,
%  \item[(iv)] $\ol\psi=0$  with  $\psi \equiv  X^{\mu}\nabla_{\mu}s    +sY -\nabla_{\mu}\Theta \nabla^{\mu}Y $,
%  \item[(v)] $\ol A_{AB} =0 = \ol A_{22}|_{N_1}= \ol A_{2A}|_{N_1}= \ol A_{11}|_{N_2}= \ol A_{1A}|_{N_2}$
%  with
% $A_{\mu\nu} \equiv \nabla_{\mu}X_{\nu} + \nabla_{\nu} X_{\mu} - 2Y g_{\mu\nu}$,
%  \item[(vi)] $\ol A_{12}=0$,
%\item[(vii)] $\ol {\nabla_{[1}A_{2]A}}|_S=0$,
%  \item[(viii)] $\ol B_{AB} =0 = \ol B_{22}|_{N_1}= \ol B_{2A}|_{N_1}= \ol B_{11}|_{N_2}= \ol B_{1A}|_{N_2}$   with
%$B_{\mu\nu}\equiv \mcL_XL_{\mu\nu} +   \nabla_{\mu}\nabla_{\nu}Y$.
% \end{enumerate}
Moreover, $\hat X=X$ and $\nabla_{\kappa}\hat X^{\kappa}= 4 Y$.
 The condition (c) suffices to be fulfilled  on $S$ and on the closure of those sets where $\tau$ is non-zero.
\end{theorem}

\subsubsection{The (proper) KID equations}

Again, we would like to replace the non-intrinsic conditions (c), (e) and $\ol \psi=0$ by conditions
which do not involve transverse derivatives of $X$ and $Y$. For the latter two this can be done as in
the light-cone case. We just note that  the ODEs  for
$\Upsilon_{N_a} $, $a=1,2$, corresponding to \eq{ODE_rho}, need to be supplemented by the boundary condition $\Upsilon_{N_a}|_S =  \ol {\partial_a Y}$.
To replace (c) one needs to take into account that, due to \eq{S110_A01},
we have
\begin{equation*}
 \ol A_{12}=0 \quad \Longleftrightarrow \quad  \ol A_{12}|_S = 0 = \ol S_{221}|_{N_1}  = \ol S_{112}|_{N_2} 
\;.
\end{equation*}
Furthermore, (b) and (c)  imply
\begin{eqnarray*}
 S_{A12}|_S &=& 2\nabla_{(A}A_{1)2} - \nabla_2 A_{1A} \,= \,  2 \nabla_{[1} A_{2]A}
\;,
\end{eqnarray*}
i.e.\ (d) can be replaced by the condition
\begin{eqnarray*}
 0\,=\, S_{A12}|_S &\equiv&   2\nabla_{A}\nabla_{1} X_{2}  - 2 R_{2 1 A}{}^{\kappa} X_{\kappa}  - 4\nabla_{(A}Y g_{1)2} + 2 \nabla_{2}Y g_{A1}
\\
 &=&  2\nabla_{A}\nabla_{1} X_{2}  - 2 R_{2 1 A}{}^{\kappa} X_{\kappa}  - 2g_{12}\nabla_{A}Y
\;.
\end{eqnarray*}

As a direct consequence of Theorem~\eq{KID_eqns_main_hyper} we end up with the following result:
\begin{theorem}
\label{KID_eqns_main_hzpe2}
Assume we have been given a $3+1$-dimensional spacetime $(\mcM, g, \Theta)$, with ($g, \Theta$) a smooth solution of the conformal field equations.
Let $\mathring X$ be a vector field and $\mathring Y$ a function defined on two transversally intersecting
null hypersurfaces $N_a \subset \mcM$, $a=1,2$, with transverse intersection along a smooth 2-dimensional submanifold $S$.
Then there exists a smooth vector field $X$ with $\ol X=\mathring X$ and $\ol{\nabla_{\kappa}X^{\kappa}} = 4\mathring Y$ satisfying the unphysical Killing equations  \eq{2_conditions} on $\mathrm{D}^+(N_1\cup N_2)$
(i.e.\ representing a Killing field of the physical spacetime)
 if and only if the KID equations are fulfilled
(we suppress the dependence of $\coneD_{i}$ on $N_a$):
 \begin{enumerate}
   \item[(i)] $  \mathring X^{\mu}\ol{\nabla_{\mu}\Theta} - \ol \Theta \mathring Y=0$,
  \item[(ii)] $ \mathring X^{\mu}\ol{\nabla_{\mu}s  }  +\ol s\mathring Y -  \ol{\nabla^{2}\Theta} \coneD_{2}\mathring Y
  -  \ol{\nabla^{A}\Theta} \coneD_{A}\mathring Y  -  \Upsilon_{N_1} \ol g^{12} \nabla_{1}\ol \Theta |_{N_1}=0$,
\\
 $ \mathring X^{\mu}\ol{\nabla_{\mu}s  }  +\ol s\mathring Y -  \ol{\nabla^{1}\Theta} \coneD_{1}\mathring Y
  -  \ol{\nabla^{A}\Theta} \coneD_{A}\mathring Y  -  \Upsilon_{N_2} \ol g^{12} \nabla_{2}\ol \Theta |_{N_2}=0$,
  \item[(iii)] $ \coneD_{(A}\mathring X_{B)}  - \mathring Y \ol g_{AB}=0$,
\\
$ \coneD_{2}\mathring X_{2}|_{N_1}  =  \coneD_{(2}\mathring X_{A)}|_{N_1}=0$,
\\
$ \coneD_{1}\mathring X_{1}|_{N_2}  =  \coneD_{(1}\mathring X_{A)}|_{N_2}=0$,
  \item[(iv)] $\coneD_{2}\coneD_{2}\mathring X_{1}  - \ol  R_{122}{}^{\kappa} \mathring X_{\kappa}  - 2\ol g_{12} \coneD_{2}\mathring Y|_{N_1} =0$,
\\
$\coneD_{1}\coneD_{1}\mathring X_{2}  - \ol  R_{211}{}^{\kappa} \mathring X_{\kappa}  - 2\ol g_{12} \coneD_{1}\mathring Y|_{N_2} =0$,
  \item[(v)] $\mathring X^{\kappa}\nabla_{\kappa}\ol L_{2i} +2 \ol L_{\kappa(2}\coneD_{i)}\mathring X^{\kappa}+   \coneD_{2}\coneD_{i}\mathring Y|_{N_1}=0$, $i=2,A$,
\\
$\mathring X^{\kappa}\nabla_{\kappa}\ol L_{1i} +2 \ol L_{\kappa(1}\coneD_{i)}\mathring X^{\kappa}+   \coneD_{1}\coneD_{i}\mathring Y|_{N_2}=0$, $i=1,A$,
  \item[(vi)] $\mathring X^{\kappa}\nabla_{\kappa}\ol L_{AB} +2 \ol L_{\kappa(A}\coneD_{B)}\mathring X^{\kappa}+   \coneD_{A}\coneD_{B}\mathring Y +\Upsilon_{N_a} \ol g^{12}\chi^{N_a}_{AB}|_{N_a}=0$, $a=1,2$,
\item[(vii)] $ \coneD_{(1}\mathring X_{2)}  - \mathring Y \ol g_{12}|_S=0$,
\item[(viii)] $2\coneD_{A}\coneD_{1} \mathring X_{2}  - 2 \ol R_{2 1 A}{}^{\kappa} \mathring X_{\kappa}  - 2\ol g_{12}\coneD_{A}\mathring Y |_S=0$,
 \end{enumerate}
where $\Upsilon_{N_1}$ is given by $\Upsilon_{N_1}|_S = \coneD_1\mathring Y$ and
\begin{eqnarray*}
 ( \partial_2 +\frac{\tau_{N_1} }{2}-\ol \Gamma^1_{12})\Upsilon_{N_1}
-\ol \Gamma^{2}_{12}\coneD_{2}\mathring Y -\ol \Gamma^{A}_{12}\coneD_{A}\mathring Y
 + \frac{1}{2}\ol g_{12}\big(\ol g^{22} \coneD_2\coneD_2 \mathring Y &&
\\
  + 2\ol g^{2A} \coneD_2\coneD_A \mathring Y+ \ol g^{AB} \coneD_A\coneD_B \mathring Y + \frac{1}{6}\mathring X^{\mu}\ol{\nabla_{\mu}R}+ \frac{1}{3} \ol R \mathring Y\big) &=&0
\;,
\end{eqnarray*}
similarly on $N_2$.

The condition (vi) is not needed on the closure of those sets on which the expansion $\tau$ is non-zero.
\end{theorem}
\begin{proof}
  Once  (i)-(viii) have been solved one uses the solutions $\mathring X$ and $\mathring Y$ as initial data for the
wave equations \eq{wave_X0} and \eq{wave_Y0}. A solution exists due to \cite{rendall}, and the rest follows from
the considerations above.
\qed
\end{proof}

\begin{remark}
{\rm
As in \cite{CP1} one could replace the condition  $ \ol g^{AB}\coneD_{(A}\mathring X_{B)}  - 2\mathring Y =0$ of (iii) by certain conditions on $S$ if one makes sure that (vi) holds regardless of the (non-)vanishing  of $\tau$.
}
\end{remark}

\begin{remark}
{\rm
Theorem~\ref{KID_eqns_main_hzpe2} can e.g.\ be applied to two null hypersurfaces  intersecting transversally  with one of them being part of $\scri^-$.
}
\end{remark}

\section{KID equations on the light-cone $\mathbf{C_{i^-}}$}
\label{sec_special_cone}

Let us  analyse now in detail the case where the initial surface is the light-cone $C_{i^-}$ with vertex at past timelike infinity $i^-$
in $3+1$-spacetime dimensions (note that this requires a vanishing cosmological constant $\lambda$).
 In particular that means %
\begin{equation}
\ol \Theta\,=\, 0
\;.
\end{equation}
That the corresponding initial value problem is well-posed for suitably prescribed data has been shown in \cite{CP2}.
Our aim is to apply Theorem~\ref{KID_eqns_main_cone2} and analyse the KID equations in this special situation.

\subsection{Gauge freedom and constraint equations}

\subsubsection{Gauge degrees of freedom}
To make computations as easy as possible it is useful to impose a convenient gauge condition. We will adopt the gauge scheme
described and justified  in \cite[Section~2.2 \& 4.1]{ttp}, where the reader is referred to for further details.
Let us start with a brief overview over the relevant gauge degrees of freedom.

The freedom to choose the conformal factor $\Theta$, regarded as an unknown in the conformal field equations \eq{conf1}-\eq{conf6},
is comprised in the freedom to prescribe the Ricci scalar $R$ and the function $\overline s$, where the latter one needs to be the restriction to $C_{i^-}$
of a smooth function, non-vanishing at $i^-$ (which ensures $\mathrm{d}\Theta|_{\scri^-}\ne 0$).

As above, we will choose adapted null coordinates $(x^0=u,x^1=r,x^A)$, $A=2,3$, on $C_{i^-}$. The freedom to choose coordinates off the cone
is reflected in the freedom to prescribe an arbitrary vector field $W^{\sigma}$ for  the \textit{$\hat g$-generalized wave-map gauge condition}
\begin{equation*}
 H^{\sigma}:=g^{\alpha\beta}(\Gamma^{\sigma}_{\alpha\beta} - \hat\Gamma^{\sigma}_{\alpha\beta}) - W^{\sigma} =0\;,
\end{equation*}
where $\hat g$ denotes some target metric.
The choice $W^{\sigma}=0$ is called \textit{wave-map gauge}.

This still leaves the freedom to parameterize the null geodesics generating $C_{i^-}$, due to which it is possible to additionally prescribe the function
\begin{equation*}
 \kappa:= \nu^0\partial_1\nu_0 - \frac{1}{2}\tau - \frac{1}{2}\nu_0(\ol g^{\mu\nu}\ol{\hat\Gamma}{}^0_{\mu\nu} + \ol W{}^0)
\;.
\end{equation*}
The choice $\kappa=0$ corresponds to an affine parameterization.
Moreover, when  $H^{\sigma}=0$ it holds that
\begin{equation*}
 \kappa = \ol \Gamma^1_{11}
\;.
\end{equation*}

\subsubsection{Constraint equations in the $(R=0, \overline s=-2, \kappa=0, \hat g=\eta)$-wave-map gauge}

Henceforth  we choose as in \cite{CP2,ttp}
\begin{equation}
 R=0\;, \quad \overline s=-2\;,\quad  W^{\sigma}=0\;, \quad \kappa=0\;, \quad \hat g=\eta
\;,
\end{equation}
where
\begin{equation*}
 \eta := -(\mathrm{d}u)^2 + 2\mathrm{d}u\mathrm{d}r + r^2 s_{AB}\mathrm{d}x^A\mathrm{d}x^B
\end{equation*}
denotes the Minkowski metric in adapted null coordinates.

Let us assume we have been given a smooth solution $(g,\Theta)$ of the conformal field equations \eq{conf1}-\eq{conf6} in the
 $(R=0, \overline s=-2, \kappa=0, \hat g=\eta)$-wave-map gauge.%
\footnote{In fact it is not necessary here to require the rescaled Weyl tensor to be regular at $i^-$.}
It is shown in \cite[Section~4]{ttp} that  then the following equations are valid on $C_{i^-}$,
\begin{eqnarray}
 &&\hspace{-3em}\overline g_{\mu\nu} = \eta_{\mu\nu} \;, \quad\overline L_{1\mu}=0\;, \quad \overline L_{AB}=\omega_{AB} \;, \quad
 \overline L_{0A} = \frac{1}{2}\tilde\nabla^B\lambda_{AB}
\;,
\label{gauge1}
\\
 &&\hspace{-3em}  \overline{\partial_0\Theta}=-2r\;, \quad  \overline{\partial_0 g_{1\mu}}=0
 \;,
\\
 &&\hspace{-3em} \tau=2/r \;, \quad  \xi_A:= -2\ol\Gamma^1_{1A}=0\;, \quad \zeta :=  2\ol g^{AB}\ol\Gamma^1_{AB} + \tau =-2/r\;,
\\
&&\hspace{-3em}  (\partial_1- r^{-1})\lambda_{AB} = - 2 \omega_{AB} \:, \quad \overline g^{AB} \lambda_{AB}= \overline g^{AB} \omega_{AB}=0
 \;,
\label{gauge4}
%\\
%  && \hspace{-3em}2(\partial_1+ r^{-1})\ol R_{00}  =\lambda^{AB} \omega_{AB}- 4r \ol d_{0101}-\tilde \nabla^A \ol R_{0A}
%\;,
%\\
%&&\hspace{-3em} 4 (\partial_1+\ 3r^{-1})\ol d_{0101} =  - \lambda^{AB} \partial_1( r^{-1}\omega_{AB})
%   - r^{-2}\tilde\nabla^A \ol R_{0A}
%  -2r^{-1} \tilde\nabla_A\tilde\nabla_{B}\omega^{AB}
%,
\end{eqnarray}
where  $\lambda_{AB}:= \overline{\partial_0 g_{AB}}=O(r^3)$.
% $\ol R_{00}=O(1$)  and $\ol d_{0101}=O(1)$.
The operator $\tilde \nabla$ denotes the Levi-Civita connection of $\tilde g :=\ol g_{AB}\mathrm{d}x^A\mathrm{d}x^B$.
The $s_{AB}$-trace-free tensor $\omega_{AB}=O(r^2)$
with $s=s_{AB}\mathrm{d}x^A\mathrm{d}x^B$ being the standard metric on $S^2$,
may be regarded as representing the free initial data in the corresponding characteristic initial value problem \cite{CP2,ttp}.

For convenience we give a list of the Christoffel symbols in adapted null coordinates on $C_{i^-}$,
which are easily obtained from \eq{gauge1}-\eq{gauge4} and the formulae in \cite[Appendix~A]{CCM2},
\begin{eqnarray*}
 &\ol \Gamma^{0}_{00}= \ol \Gamma^{\mu}_{01}= \ol \Gamma^{\mu}_{11}
= \ol \Gamma^{0}_{0A}= \ol \Gamma^{0}_{1A}= \ol \Gamma^{1}_{0A}= \ol \Gamma^{1}_{1A}
=0\;,&
  \\
&\ol \Gamma^{1}_{00} =\frac{1}{2}\ol{\partial_0 g_{00}}\;, \quad \ol \Gamma^{C}_{00} =\ol g^{CD}\ol{\partial_0 g_{0D}} \;, \quad \ol \Gamma^0_{AB} = -r^{-1}\ol g_{AB} \;, \quad  \ol \Gamma^C_{1A} = r^{-1}\delta_A{}^C
\;,&
\\
&\ol\Gamma^C_{0A} = \frac{1}{2}\lambda_A{}^C
\;, \quad   \ol \Gamma^1_{AB} = -r^{-1}\ol g_{AB} - \frac{1}{2}\lambda_{AB}  \;, \quad \ol\Gamma^C_{AB} = \tilde\Gamma^C_{AB} = S^C_{AB}
\;.&
\end{eqnarray*}

%For the components of the Riemann tensor one finds \cite{ttp}
%%
%\begin{eqnarray}
% &\overline R_{0101} \,=\,
%     0
%  \;, \quad
%  \overline R_{011A} \,=\,
%  0
%  \;, \quad
%  \overline R_{01AB} \,=\,
%  0
%  \;, \quad
%  \overline R_{1A1B} \,=\,
%  0
%  \;,&
% \\
%&  \overline R_{010A} \,=\, -\frac{1}{2}\ol R_{0A}
%  \;,
%\quad
%\ol R_{0A0B} \,=\, \frac{1}{2}\ol g_{AB}\ol R_{00} - \frac{1}{2}\omega_{AB}
%  \;.&
%\end{eqnarray}

\subsection{Analysis of the KID equations}

\subsubsection{The conditions $\ol\phi=0$, $\ol{\psi}^{\mathrm{intr}}=0$, $\ol A_{ij}=0$ and $\ol S_{110}=0$}

With $\ol\Theta=0$ and $ \overline{\partial_0\Theta}=-2r$ it immediately follows that
\begin{eqnarray}
 \ol\phi=0 \quad \Longleftrightarrow \quad \mathring X^0=0
 \;,
\end{eqnarray}
i.e.\ any vector field satisfying the unphysical Killing equations necessarily needs to be tangent to $C_{i^-}$.

Taking further into account that $\ol s=-2$ and $\nu_0=1$ we obtain
(recall that  $\ol{\psi}^{\mathrm{intr}}$ has been defined in Theorem~\ref{KID_eqns_main_cone2})
\begin{eqnarray}
  \ol{ \psi}^{\mathrm{intr}}=0 \quad \Longleftrightarrow \quad (\partial_1-r^{-1})\mathring Y = 0   \quad \Longleftrightarrow \quad  \mathring Y=c(x^A)r
 \;,
\end{eqnarray}
for some angle-dependent function $c$.
The condition  $\ol A_{11}=0$ is then automatically fulfilled.
Furthermore, one readily checks that (we denote by $\mcD$ the Levi-Civita connection associated to the standard metric on $S^2$)
\begin{eqnarray}
\ol A_{1A} =0 & \Longleftrightarrow &  \partial_1 \mathring X^A =0\quad \Longleftrightarrow \quad \mathring  X^A= d^A(x^B)
 \;,
\\
 \ol g^{AB} \ol A_{AB} =0 & \Longleftrightarrow &  \mathring X^1 =  - \frac{1}{2}r\mcD_Ad^A +  cr^2
 \;,
\\
 \breve{ \ol A}_{AB} =0 & \Longleftrightarrow &  \text{$d^A$ is a conformal Killing field on $(S^2,s_{AB})$}
 \;.
\end{eqnarray}
Here and in what follows $\breve{.}$ denotes the $s_{AB}$- (equivalently the $\ol g_{AB}$-) trace-free part of the corresponding rank-2 tensor field.

Since $\tau = 2/r>0$ the condition $\ol S_{110}=0$ holds automatically for all $r>0$.

\subsubsection{The conditions $\ol B_{1i}=0$ and $\ol{ B}^{\mathrm{intr}}_{AB}=0$}

First we solve \eq{thm_rho_eqn} for $\Upsilon$, which in our gauge becomes %on $C_{i^-}$
\begin{eqnarray}
 ( \partial_1 + r^{-1})\Upsilon
  + \frac{1}{2}r^{-2}\Delta_s \mathring Y  +r^{-1}\partial_1 \mathring Y   =0
\;,
\label{special_Upsilon}
\end{eqnarray}
where we have set $\Delta_s := s^{AB}\mcD_A\mcD_B$.
With $\mathring Y = cr$ and $\Upsilon=O(1)$ we obtain as the unique solution of  \eq{special_Upsilon}
\begin{equation}
 \Upsilon =  - \frac{1}{2}(\Delta_s +2)c
 \;.
\end{equation}
For $\ol B_{1i}$ we find
\begin{eqnarray*}
 \ol B_{11} &\equiv &  \mathring X^{\kappa}\nabla_{\kappa}\ol L_{11} +2 \ol L_{\kappa(1}\coneD_{1)}\mathring X^{\kappa}+   \coneD_{1}\coneD_{1}\mathring Y
\\
 &=& 0
\;,
\\
\ol B_{1A} &\equiv &  \mathring X^{\kappa}\nabla_{\kappa}\ol L_{1A} +2 \ol L_{\kappa(1}\coneD_{A)}\mathring X^{\kappa}+   \coneD_{1}\coneD_{A}\mathring Y
\\
&=& \omega_{ AB}\partial_{1}\mathring X^{B} +  \partial_{A}  (\partial_1 -r^{-1})\mathring Y
\\
&=& 0
\;.
\end{eqnarray*}
without any further restrictions on $\mathring X$, $\mathring Y$ or the initial data $\omega_{AB}$.
It remains to determine  $\ol{B}^{\mathrm{intr}}_{AB}$,
\begin{eqnarray*}
 \ol{ B}^{\mathrm{intr}}_{AB} &=&   \mathring X^{1}\nabla_{1}\ol L_{AB}  +  \mathring X^{C}\nabla_{C}\ol L_{AB}
+2 \ol L_{0(A}\coneD_{B)}\mathring X^{0} +2 \ol L_{C(A}\coneD_{B)}\mathring X^{C}
\\
&&
+   \coneD_{A}\coneD_{B}\mathring Y + r^{-1}  \ol g_{AB} \Upsilon
\\
&=&    \mathring X^{1}\partial_1\omega_{AB}  +  \mathring X^{C}\tilde \nabla_{C}\omega_{AB}  +2 \omega_{C(A}\tilde\nabla_{B)}\mathring X^{C}
\\
&&
+  \tilde\nabla_A\tilde\nabla_B\mathring Y   + \frac{1}{2}\lambda_{AB}\partial_1\mathring Y     + r^{-1} \ol g_{AB}(\partial_1\mathring Y + \Upsilon)
\;.
\end{eqnarray*}
We first compute its trace,
\begin{eqnarray*}
 \ol g^{AB}\ol{ B}^{\mathrm{intr}}_{AB}
&=&  2 \omega^{AB}(\tilde\nabla_{A}\mathring X_B)\breve{}  + \Delta_{\tilde g}\mathring Y     + 2r^{-1} \partial_1\mathring Y + 2r^{-1} \Upsilon
\\
 &=& 0
\;,
\end{eqnarray*}
again without any further restrictions.
For its traceless part we find
\begin{eqnarray*}
 \breve{\ol{ B}}{}^{\mathrm{intr}}_{AB} &=& \mathring X^{1}\partial_1\omega_{AB}  +  \mathring X^{C}\tilde \nabla_{C}\omega_{AB}
  +2 \omega_{C(A}\tilde\nabla_{B)}\mathring X^{C}
 - \ol g_{AB} \omega^{CD}(\tilde\nabla_{C}\mathring X_{D})\breve{}
\\
&&
+ ( \tilde\nabla_A\tilde\nabla_B\mathring Y )\breve{}  + \frac{1}{2}\lambda_{AB}\partial_1\mathring Y
\\
&=&   \mcL_d\omega_{AB} - \frac{1}{2}r\partial_1\omega_{AB}\mcD_Cd^C   +  cr^2\partial_1\omega_{AB}
  + \frac{1}{2}c\lambda_{AB} + r( \mcD_A\mcD_Bc )\breve{}
\;.
\end{eqnarray*}
Recall that regularity of the metric requires $\omega_{AB}=O(r^2)$ and $\lambda_{AB}=O(r^3)$, in particular $\mcL_d\omega_{AB}=O(r^2)$.
Hence $\breve{\ol{B}}{}^{\mathrm{intr}}_{AB}=0$ if and only if
\begin{eqnarray}
 &\text{$\mathring\nabla_A c$ is a conformal Killing field on $(S^2,s_{AB})$},&
\\
 & \mcL_d\omega_{AB} - \frac{1}{2}r\partial_1\omega_{AB}\mathring \nabla_Cd^C   +  cr^2\partial_1\omega_{AB}
  + \frac{1}{2}c\lambda_{AB}  =0
\;.&
\end{eqnarray}

\subsubsection{Summary}

By way of summary the conditions (i)-(vi) in Theorem~\ref{KID_eqns_main_cone2} hold if and only if
\begin{eqnarray}
 \mathring X^0 &=& 0\;,
\label{cond_X0}
\\
  \mathring X^A &=& d^A\;,
\label{cond_XA}
\\
 \mathring X^1 &=& -\frac{1}{2} r\mcD_A d^A + cr^2\;,
\label{cond_X1}
\\
 \mathring Y &=& cr \;,
\label{cond_Y}
\end{eqnarray}
such that
\begin{eqnarray}
 & \text{$\mcD_A c$ and $d_A$ are conformal Killing fields on $(S^2,s_{AB})$},&
\label{cond_conf}
\\
  &\mcL_d\omega_{AB} - \frac{1}{2}r\mcD_Cd^C\partial_1\omega_{AB}   +  cr^2\partial_1\omega_{AB}
  + \frac{1}{2}c\lambda_{AB}  =0
\;.&
\label{reduced_KID}
\end{eqnarray}
%
%Note that the conformal Killing equation for $d^A$ implies the relation
%
%\begin{eqnarray}
% (\Delta_s+1)d_A &=& 0
% \;.
%\end{eqnarray}

In Section~\ref{extendability} we have shown that solutions of the KID equations are restrictions to the light-cone of smooth spacetime fields.
On $C_{i^-}$ this turns out to be a trivial issue anyway:
The candidate fields satisfying \eq{cond_X0}-\eq{cond_conf} are explicitly known%
\footnote{The function $c$ satisfies the equation $\mcD_A(\Delta_s+2) c=0$ and can thus be written as  linear combination of $\ell=0,1$ spherical harmonics.
%The vector field $d^A$ satisfies $(\Delta_s+1)d_A = 0$ and is thus an eigenform of the Laplacian.
Conformal Killing vector fields on the round 2-sphere are discussed in  Appendix~\ref{app_conformal}.
}
and coincide
\textit{independently of the choice of initial data} $\omega_{AB}$, with the restriction to $C_{i^-}$ of the Minkowskian Killing vector fields.
%This observation immediately guarantees that all candidate fields satisfying \eq{cond_X0}-\eq{cond_conf} extend to smooth spacetime fields.

While in the Minkowski case $\omega_{AB}=0$ every candidate field does extend to a Killing vector field, equation \eq{reduced_KID} provides an obstruction equation for non-flat data.
We call \eq{reduced_KID} the \textit{reduced KID equations}.
%i.e.\ condition (vii) in Theorem~\ref{KID_eqns_main_cone2} will always be satisfied by any solution of
%the KID equations \eq{cond_X0}-\eq{reduced_KID}.

As a corollary of Theorem~\ref{KID_eqns_main_cone2} we obtain:
\begin{theorem}
\label{KID_eqns_main_cone_infinity}
Assume that  we have been given a $3+1$-dimensional ``unphysical'' spacetime $(\mcM, g, \Theta)$
which contains a regular $C_{i^-}$-cone (the cosmological constant $\lambda$ thus needs to vanish)
and where ($g, \Theta$) is a smooth solution of the conformal field equations  in the $(R=0, \overline s=-2, \kappa=0, \hat g=\eta)$-wave-map gauge.
Then there exists a smooth vector field $X$
%with $\ol X=\mathring X$ and $\ol{\nabla_{\kappa}X^{\kappa}} = \frac{1}{4}\mathring Y$, $\mathring X$ and $\mathring Y$ given by \eq{cond_X0}-\eq{cond_Y},
satisfying the unphysical Killing equations \eq{2_conditions} on $\mathrm{D}^+(C_{i^-})$
(i.e.\  representing a Killing field of the physical spacetime)
if and only if
there exist a function $c$ and a vector field $d^A$ on $S^2$ with $\mcD_A c$ and $d_A$
conformal Killing fields on $(S^2,s_{AB})$ such that the reduced KID equations
\begin{equation}
   \mcL_d\omega_{AB} - \frac{1}{2}r\partial_1\omega_{AB}\mcD_Cd^C   +  cr^2\partial_1\omega_{AB}
  + \frac{1}{2}c\lambda_{AB}  =0
\label{thm_reduced_KID}
\end{equation}
 are satisfied on $C_{i^-}$ (recall that $\lambda_{AB}$ is the unique solution of $(\partial_1-r^{-1})\lambda_{AB}=-2\omega_{AB}$
with $\lambda_{AB}=O(r^3)$).

The Killing field satisfies
\begin{eqnarray}
 \ol X^0=0\;, \quad \ol X^A=d^A\;, \quad \ol X^1=-\frac{1}{2}r\mcD_Ad^A +cr^2\;,\quad \ol{\nabla_{\mu}X^{\mu}} = 4cr\;.
\end{eqnarray}
%If there is more than one such pair $(c,d^A)$, each of them defines a vector field satisfying the unphysical Killing equations.
\end{theorem}

\begin{remark}
{\rm
The reduced KID equations \eq{thm_reduced_KID} can be replaced by one of their equivalents (i)-(iii) in Lemma~\ref{lemma_red_KIDs}.
}
\end{remark}

\subsection{Analysis of the reduced KID equations}

\subsubsection{Equivalent representations of the reduced KID equations}

We provide some alternative formulations of the reduced KID equations.
\begin{lemma}
\label{lemma_red_KIDs}
 The reduced KID equations  \eq{thm_reduced_KID} are equivalent to each of the following equations:
\begin{enumerate}
\item[(i)] $ \mcL_d \lambda_{AB} - ( \frac{1}{2}r\mcD_C d^C - cr^2)\partial_1\lambda_{AB}
+(\frac{1}{2} \mcD_C d^C  - 2cr)\lambda_{AB}=0$,
\item[(ii)]  $(\partial_1-r^{-1})\mcL_d\omega_{AB} - \frac{1}{2}r\partial^2_{11}\omega_{AB}\mcD_Cd^C
 + cr^2\partial_1(\partial_1 + r^{-1})\omega_{AB}=0$,
\item[(iii)] $2 \mcL_d \ol L_{0A}
  + (1  -r\partial_1)\ol L_{0A} \mcD_{B}d^B
 +r \omega_{A}{}^C\mcD_C \mcD_B d^B
+  2cr^2 \partial_1\ol L_{0A}
-(2\omega_{AB}+r^{-1}\lambda_{AB}) \mcD^B c = 0$ (recall that $\ol L_{0A} = \frac{1}{2}\tilde\nabla_B\lambda_A{}^B$).
\end{enumerate}
\end{lemma}

\begin{proof}
(i) Applying $(\partial_1-r^{-1})$ to equation (i) yields \eq{thm_reduced_KID}, equivalence follows from regularity.
\\
(ii) Applying $(\partial_1-r^{-1})$ to \eq{thm_reduced_KID} yields equation (ii), equivalence follows from regularity.
\\
(iii) We use the fact that on $(S^2, s_{AB}$) the equations $w_{AB}=0$ and $\mcD^Bw_{AB}=0$ with $w_{AB}$ trace-free,
are equivalent: Taking the divergence of (i) and invoking  the conformal Killing equation for $d^A$
then completes the equivalence proof.
\qed
\end{proof}

Both $\omega_{AB}$ or $\lambda_{AB}$ may be regarded as the freely prescribable initial data.
So (i) and (ii) in Lemma~\ref{lemma_red_KIDs} provide formulations of the reduced KID equations which involve exclusively explicitly known quantities
for all admissible initial data.
In the case of an ordinary cone, treated in \cite{CP1}, this was not possible: For generic KIDs there,
neither the candidate fields
nor all the relevant metric components can be computed analytically.

\subsubsection{Some special cases}

We finish with a brief discussion of some special cases:
There exists a vector field $X$ satisfying the unphysical Killing equations \eq{2_conditions} on $\mathrm{D}^+(C_{i^-})$
with
\begin{enumerate}
\item $\ol{\nabla_{\mu}X^{\mu}}=0 %\quad\Longleftrightarrow \quad \mathring Y=0
\enspace\Longleftrightarrow \enspace \exists$ a conformal Killing vector field $d^A$ on  $(S^2,s_{AB})$  with
 $ \mcL_d\omega_{AB} = \frac{1}{2}r\partial_1\omega_{AB}\mcD_Cd^C$,
\item $\ol X^1=0 \enspace\Longleftrightarrow \enspace\exists$ a  Killing vector field $d^A$ on  $(S^2,s_{AB})$  with
 $ \mcL_d\omega_{AB} = 0$,
\item $\ol X^A =0 \enspace\Longleftrightarrow \enspace
%  r^2\partial_1\omega_{AB}
%  + \frac{1}{2}\lambda_{AB}  =0 \enspace \overset{\eq{gauge4}}{\Longleftrightarrow} \enspace
 \partial_1(\partial_1 + r^{-1})\omega_{AB}=0 \quad \overset{\omega_{AB}=O(r^2)}{\Longleftrightarrow} \quad  \omega_{AB}=0$
\\
($\omega_{AB}\equiv \breve{\ol L}_{AB}=O(r^2)$ is a necessary condition on the Schouten tensor to be regular at $i^-$).
\end{enumerate}

The third case shows that the property $\ol X^A=0$ is compatible only with the Minkowski case (supposing that $i^-$ is a regular point). In the non-flat case
any non-trivial vector field satisfying the unphysical Killing equations has a non-trivial component $\ol X^A = d^A \not\equiv 0$.
Since
\begin{equation*}
 \ol g_{\mu\nu} \ol X^{\mu}\ol X^{\nu} = \ol g_{AB} \ol X^{A}\ol X^{B} = r^2 s_{AB} d^A d^B
 \;,
\end{equation*}
we see that there are no non-trivial vector fields satisfying the unphysical Killing equations which are null on $C_{i^-}$.
%(of course $d^A$ has zeros, so it will be null along certain single null geodesics emanating from $i^-$).
To put it differently, possibly apart from certain directions determined by the zeros of $d^A$, any isometry of a non-flat, asymptotically flat vacuum spacetime is necessarily spacelike sufficiently close to $\scri^-$.
This leads to the following version of a classical result of Lichnerowicz~\cite{lich}:

\begin{theorem}
Minkowski spacetime is the only
stationary vacuum spacetime which admits a regular $C_{i^-}$-cone.
\end{theorem}

%\tim{regular $C_{i^-}$-cone = geodesically complete and asymptotically flat as required by Lichnerowicz? \\ -- \\ ptcr: modern lichnerowicz is formulated in terms of complete initial data, and then not even asymptotic flatness is needed in some versions}
%\ptcr{from the above one should also get that the only possibly symmetry is rotations, unless the space-time is flat? what about boosts?
%\\ -- \\
%tim: this is not clear to me \\ -- \\ remind me to discuss this in Vienna by mid October}

\subsubsection{Structure of the solution space}

Let $X$ and $\hat X$ be two distinct non-trivial solutions of the unphysical Killing equations \eq{2_conditions}.
Since solutions of these equations form a Lie algebra,
$\hat{\hat X}:= [X,\hat X]$ is another,  possibly trivial, solution.
We have
\begin{eqnarray*}
\ol{ \hat{\hat X}}{}^0\,=\, \ol{[X,\hat X]}{}^{0} &=& 0
\;,
\\
  \ol{\hat{\hat X}}{}^A\,=\, \ol{[X,\hat X]}{}^{A} &=& [d,\hat d]^A
\;,
\\
 \ol{\hat{\hat X}}{}^1\,=\,  \ol{[X,\hat X]}{}^{1} &=&  -\frac{1}{2} r\mcD_B[d,\hat d]^B
+r^2(d^{B}\mcD_{B} \hat c -\hat d^{B}\mcD_{B} c +\frac{1}{2} c\mcD_B \hat d^B   - \frac{1}{2}  \hat c\mcD_B d^B )
\;.
\end{eqnarray*}
Hence, by their derivation, the reduced KID equations are fulfilled  with
\begin{eqnarray*}
 \hat{\hat d}^A &=& [d,\hat d]^A
\;,
\\
\hat{\hat c} &=& d^{B}\mcD_{B} \hat c -\hat d^{B}\mcD_{B} c +\frac{1}{2} c\mcD_B \hat d^B   - \frac{1}{2}  \hat c\mcD_B d^B
\;,
\end{eqnarray*}
and $\hat{\hat d}^A$ and $\mcD_A \hat{\hat c}$ are conformal Killing fields on the standard 2-sphere.
% since $\mcD_A c$, $\mcD_A\hat c$, $d^A$ and $\hat d^A$ are.)
Indeed, via the relation $ \mcL_{[d,\hat d]}\lambda_{AB} =[ \mcL_d,\mcL_{\hat d}]\lambda_{AB}$, this can be straightforwardly checked.
%Let $d=0$  if $d^A$ is a Killing field, and let $d$ be the $\ell=1$-spherical harmonics for which $d_A=\mcD_A d$ otherwise, and define the function $\hat d$ in an analogous way.
%Then
%Moreover,
%%
%\begin{eqnarray*}
%\mcD_A\hat{\hat c} &=&[d,\mcD\hat c]_A -[\hat d,\mcD c]_A
%   +\frac{1}{2} (c\mcD_A - \mcD_A c) \mcD_B \hat d^B
%\\
% &&   + \frac{1}{2} (\mcD_A \hat c -  \hat c\mcD_A)\mcD_B d^B
%\\
% &=&[d,\mcD\hat c]_A -[\hat d,\mcD c]_A
%  - c  \hat d_A -\frac{1}{2}\mcD_B \hat d^B \mcD_A c
%\\
% &&   +\frac{1}{2}\mcD_B d^B\mcD_A \hat c + \hat c d_A
% \;.
%\end{eqnarray*}
%%%%%%%%%%
%%%%%%%%%%%
%\begin{eqnarray*}
% \mcL_{[d,\hat d]}\lambda_{AB} &=&[ \mcL_d,\mcL_{\hat d}]\lambda_{AB}
%\;,
%\end{eqnarray*}
%%
%one shows that the reduced KID equations are satisfied,
%%
%\begin{eqnarray*}
%&&\hspace{-3em} \mcL_{\hat{\hat d}} \lambda_{AB} - ( \frac{1}{2}r\mcD_C \hat{\hat d}^C -\hat{\hat c}r^2)\partial_1\lambda_{AB}
%+(\frac{1}{2} \mcD_C \hat{\hat d}^C  - 2\hat {\hat c}r)\lambda_{AB}
%\\
%&=&
%\frac{1}{2}r\mcD_C \hat d^C\partial_1\mcL_d\lambda_{AB}
%-    \frac{1}{2}r\mcD_C d^C\partial_1\mcL_{\hat d}\lambda_{AB}
% -\frac{1}{2} \mcD_C \hat d^C \mcL_d\lambda_{AB}
%\\
%&&
%+\frac{1}{2} \mcD_C d^C\mcL_{\hat d}\lambda_{AB}
%+ (c-\hat c)(r^2\partial_1-2r)\mcL_d\lambda_{AB}
%\\
%&& + (\hat{\hat c}+ \hat d^D\mcD_Dc   - d^D\mcD_D\hat c) (r^2\partial_1\lambda_{AB}  -2r\lambda_{AB})
%%%%%%%%%%%%%%%%%%%%%%
%\\
%&=&
%(r^2\partial_1\lambda_{AB}- 2r\lambda_{AB} ) (\hat{\hat c} + \hat d^D\mcD_Dc   - d^D\mcD_D\hat c + \frac{1}{2}\hat c\mcD_D d^D- \frac{1}{2}c\mcD_D \hat %d^D  )
%\\
%&=& 0
%\;.
%\end{eqnarray*}
We refer the reader to Appendix~\ref{app_conformal} where the conformal Killing fields on the standard 2-sphere are explicitly given.

Let us consider for the moment flat initial data $\lambda_{AB}=0$ which generate Minkowski spacetime.
Then one has 10 independent isometries:
\begin{itemize}
\item The four \textit{translations} are generated by the tuples $(c,d^A=0)$ with $c$ being a spherical harmonic function of degree $\ell=0$ or 1.
\item The three \textit{rotations} are generated by the tuples $(c=0,d^A)$ with $d^A$ being a Killing field on $(S^2,s_{AB}\mathrm{d}x^A\mathrm{d}x^B)$.
\item The three \textit{boosts} are generated by the tuples $(c=0,d^A=\mcD^A f)$ with $f$ being  a spherical harmonic function of degree $\ell=1$.
\end{itemize}

We have already seen above that translations (in the above sense) cannot exists in the non-flat case $\lambda_{AB}\ne 0$ if the Schouten tensor
is assumed to be regular at $i^-$.

\begin{proposition}
 Minkowski spacetime is the only spacetime with a regular $C_{i^-}$-cone which admits translational Killing vector fields.
%\tim{=no translations? \\ -- \\ ptc: what about the C metrics of Ashtekar and Dray? or more generally the polarised boost rotation symmetric metrics of Bicak and Schmidt? I am also confused by the question, where have the translations vanished from the picture}
\end{proposition}

This is linked with another observation: Since, in the non-flat case,  any non-trivial Killing field of the physical spacetime (i.e.\ a vector field satisfying the unphysical Killing equations)
%\tim{spacetime Killing field = solution of the unphysical Killing eqns}
has a non-trivial $d^A$, for a given $d^A$ there can be at most one $c$ such that $(c,d^A)$ solves the reduced KID equations. Now the standard 2-sphere admits 6 independent conformal Killing vector fields $d^A$. We thus have:

\begin{proposition}
 Any non-flat spacetime with a regular $C_{i^-}$-cone admits at most 6 independent Killing vector fields.
\end{proposition}

Now let us assume that there are two distinct rotations,
i.e.\ 2 Killing fields $d^{(1)}$ and $d^{(2)}$ on $(S^2,s_{AB}\mathrm{d}x^A\mathrm{d}x^B)$
such that $(c=0,d=d^{(i)})$, $i=1,2$, solves the reduced KID equations.
Then $(c=0,d=d^{(3)})$ with $d^{(3)}=[d^{(1)},d^{(2)}]$ provides another independent, non-trivial solution of the reduced KID equations.
Altogether we have
\begin{equation*}
 \mcL_{d^{(i)}}\lambda_{AB}=0\;, \quad i=1,2,3 \quad \Longrightarrow \quad \lambda_{AB}\propto s_{AB}  \quad \Longrightarrow \quad \lambda_{AB}=0
\;,
\end{equation*}
since $\lambda_{AB}$ is trace-free. This recovers the well-known fact that two rotational symmetries imply Minkowski spacetime.

\vspace{1.2em}
\noindent {\textbf {Acknowledgements}}
It is a pleasure to thank my advisor Piotr T. Chru\'sciel for various valuable comments as well as for reading a first draft of this
article.
Supported in part by the  Austrian Science Fund (FWF): P 24170-N16.

\newpage

\appendix

\section{Fuchsian ODEs}
\label{app_Fuchsian}

As we have not been able to find an adequate reference,
we state and prove here a key result about Fuchsian ODEs which is used in our work.

\begin{lemma}
\label{lemma_Fuchsian}
Let $a>0$, and for $r\in (0,a)$
consider a first-order ODE-system of the form
\begin{equation}
 \partial_r\phi = r^{-1}A \phi + M(r)\phi\;,
\label{Fuchsian_system}
\end{equation}
for a set of fields $\phi=(\phi^I)$, $I=1,\dots,N$,
where $A$ is an
 $N\times N$-matrix,
and where $M(r)$ is a continuous
map
on $[0,a)$  with values in $N\times N$-matrices
which satisfies $r\|M(r)\|_{op}=o(1)$. Let $\lambda$   denote the smallest number so that
$$
 \braket{\phi, A \phi } \le \lambda \|\phi\|^2
 \;.
$$
Suppose that there exists $\epsilon>0$ such that
%Then the only solution of \eq{Fuchsian_system} for which
$$ \phi =O(r^{\lambda+\epsilon})
\;.
$$
%.
Then
$$
 \phi\equiv 0
 \;.
$$
%.
\end{lemma}

\begin{proof}
The proof is done  by a simple energy estimate.
Set
\begin{eqnarray*}
 \braket{\phi,\psi }:= \sum_I\phi^I\psi^I\;, \quad \|\phi\|^2:=\braket{\phi,\phi}
\;,
\end{eqnarray*}
%
%and denote by $\|\cdot\|_{op}$ the operator norm,
then for any $k\in\mathbb{R}$
\begin{eqnarray*}
 \partial_r(r^{-2k}\|\phi\|^2) &=&  2r^{-2k}\phi\partial_r\phi -2kr^{-2k-1}\|\phi\|^2
\\
&=&2  r^{-2k-1}(\braket{\phi, A \phi }+r\braket{\phi, M(r)\phi}-k\|\phi\|^2)
%\\
%&\leq & 2  r^{-2k-1}[(\lambda + 1)\|\phi\|^2+r|\braket{\phi, M(r)\phi}|-k\|\phi\|^2]
\\
&\leq & 2  r^{-2k-1}(\lambda  -k+r\|M(r)\|_{op} )\|\phi\|^2
\;.
\end{eqnarray*}
 Applying $\int_{r_0}^r$  yields (assume $r_0<r$)
\begin{eqnarray*}
 r^{-2k}\|\phi(r)\|^2
&\leq&  r_0^{-2k}\|\phi(r_0)\|^2 + 2 \int_{r_0}^r(\lambda -k+\tilde r\|M(\tilde r)\|_{op} ) \tilde r^{-2k-1}\|\phi\|^2\,\mathrm{d}\tilde r
\\
&\leq&  r_0^{-2k}\|\phi(r_0)\|^2 + 2\Big(\lambda-k+\sup_{0<\tilde r< r}(\tilde r\|M(\tilde r)\|_{op}) \Big) \int_{r_0}^r \tilde r^{-2k-1}\|\phi\|^2\,\mathrm{d}\tilde r
\;.
\end{eqnarray*}
Due to our assumption $\phi=O(r^{\lambda + \varepsilon})$ any $\lambda < k_0 < \lambda +\varepsilon $ satisfies $r^{-2k_0}\|\phi\|^2=O(r^{2\delta})$, where $\delta:= \lambda - k_0 +\varepsilon>0$.
We then take the limit $r_0\rightarrow 0$,
\begin{eqnarray*}
 r^{-2k_0}\|\phi(r)\|^2
&\leq&   2\Big(\lambda -k_0+\sup_{0<\tilde r< r}(\tilde r\|M(\tilde r)\|_{op}) \Big) \int_{0}^r \tilde r^{-2k_0-1}\|\phi\|^2\,\mathrm{d}\tilde r
\\
&\leq& 0 \quad \text{for sufficiently small $r$}
\;.
\end{eqnarray*}
Thus $\phi$ vanishes for small $r$, but then it needs to vanish for all $r$.
\qed
\end{proof}

%\begin{lemma}
%\label{app_aux}
%Any solution $\tilde\phi \in C([0,a))$ of \eq{Fuchsian_system1b}  is smooth at $r=0$.
%\end{lemma}

%\begin{proof}
% From \eq{Fuchsian_system1b} we obtain
%%
%\begin{eqnarray*}
%\tilde\phi(r) =   e^{-A\log\frac{r}{r_0}  - \int_{r_0}^r M(\tilde r)\mathrm{d}\tilde r} \phi(r_0)
%\end{eqnarray*}
%\end{proof}

\section{Conformal Killing fields on the round 2-sphere}
\label{app_conformal}

We consider the 2-sphere equipped with the standard metric
\begin{eqnarray*}
 s=s_{AB}\mathrm{d}x^A\mathrm{d}x^B = \mathrm{d}\Theta^2 + \sin^2\Theta\mathrm{d}\varphi^2
\;.
\end{eqnarray*}
It admits the maximal number of independent conformal Killing vector fields, which is 6.
There are three independent Killing vector fields,
\begin{eqnarray*}
 K_{(1)}&=&\partial_{\varphi}\;,
\\
K_{(2)} &=& \sin\varphi \partial_{\Theta} +\cot\Theta\cos\varphi\partial_{\varphi}\;,
\\
K_{(3)} &=& \cos\varphi \partial_{\Theta} -\cot\Theta\sin\varphi\partial_{\varphi}\;,
\end{eqnarray*}
and three independent conformal Killing fields which are not Killing fields,
\begin{eqnarray*}
 C_{(1)} &=& \sin\Theta \partial_{\Theta} \;,
\\
C_{(2)} &=& \cos\Theta\cos\varphi \partial_{\Theta} -\sin^{-1}\Theta\sin\varphi\partial_{\varphi}
\;,
\\
C_{(3)} &=& \cos\Theta\sin\varphi \partial_{\Theta} +\sin^{-1}\Theta\cos\varphi\partial_{\varphi}
\;.
\end{eqnarray*}
All the $C_{(i)}$'s turn out to be gradients of $\ell=1$-spherical harmonics,
\begin{eqnarray*}
 C^A_{(1)}&=&\mcD^A c_{(1)}\;, \quad \text{where} \quad c_{(1)}\,=\,\cos\Theta\;,
\\
 C^A_{(2)}&=&\mcD^A c_{(2)}\;, \quad \text{where} \quad c_{(2)}\,=\,\sin\Theta\cos\varphi\;,
\\
 C^A_{(3)}&=&\mcD^A c_{(3)}\;, \quad \text{where} \quad c_{(3)}\,=\,\sin\Theta\sin\varphi\;,
\end{eqnarray*}
Moreover,
\begin{eqnarray*}
 \mcD_AC_{(i)}^A \,=\, \mcD_A\mcD^Ac_{(i)} \,=\, -2c_{(i)}\;, \quad i=1,2,3\;.
\end{eqnarray*}
The conformal Killing fields satisfy the commutation relations
\begin{eqnarray*}
 \big[ K_{(i)}, K_{(j)}\big] &=& \varepsilon_{ijk}K_{(k)}\;,
\\
 \big[ C_{(i)}, C_{(j)} \big] &=& - \varepsilon_{ijk}K_{(k)}\;,
\\
 \big[ K_{(i)}, C_{(j)}\big]    &=&\varepsilon_{ijk}C_{(k)}\;,
\end{eqnarray*}
i.e.\ they form a Lie algebra isomorphic to the Lie algebra  $so(3,1)$  of the Lorenz group in 4 dimensions.
The Killing fields form a Lie subalgebra.

\end{document}